\documentclass[twocolumn]{autart}    

\usepackage{amsmath}
\usepackage{amssymb}
\usepackage{amsfonts}
\usepackage{fleqn}
\usepackage{mathtools,array,nicematrix}
\usepackage{tikz}
\usepackage{bm}
\usetikzlibrary{decorations.pathreplacing}
\usetikzlibrary{fit}
\usepackage{algorithmic}
\usepackage{algorithm}
\usepackage{cases}
\usepackage{subfig}
\usepackage{textcomp}
\usepackage{stfloats}
\usepackage{url}
\usepackage{verbatim}
\usepackage{cite}
\usepackage{mathrsfs}
\usepackage{mathptmx}
\usepackage{array}
\usepackage{arydshln}
\usepackage{multirow}
\usepackage{booktabs}
\usepackage{threeparttable}
\usepackage[dvipsnames]{xcolor}
\usetikzlibrary{backgrounds, arrows.meta, positioning}
\pgfdeclarelayer{background}
\pgfsetlayers{background,main}
\usepackage{cancel}
\usepackage{enumitem}
\definecolor{yourcolor}{RGB}{128, 0, 0}

\usepackage{bbm}
\usepackage{iitem}
\usepackage{mathtools,array,nicematrix}
\usepackage{graphicx}          
\usepackage{color, xcolor}

\begin{document}

 \allowdisplaybreaks[4]

\newtheorem{theorem}{Theorem}
\newtheorem{definition}{Definition}
\newtheorem{assumption}{Assumption}
\newtheorem{lemma}{Lemma}
\newtheorem{example}{Example}
\newtheorem{remark}{Remark}
\newtheorem{corollary}{Corollary}
\newenvironment{proof}{{\noindent\bf Proof.}\ }{\hfill $\square$\par}

\begin{frontmatter}

\title{\textcolor{black}{When classical predictors fail: an exact network outflow predictor for heterogeneous multi-agent systems with communication delays}} 
   

\thanks[footnoteinfo]{The work of Qin Fang and Yang Zhu was partially supported by Zhejiang R\&D Program (Grant No. 2025C01022), National Natural Science Foundation of China (Grant No. 62303410), Zhejiang Provincial Natural Science Foundation of China (Grant No. LQ23F030014), and Open Research Project of State Key Laboratory of Industrial Control Technology, China (Grant No. ICT2025B78, ICT2025B09). The work of Mamadou Diagne was funded by the NSF CAREER Award CMMI-2302030 and  the NSF grant CMMI-2222250.  Corresponding author: Yang Zhu.}

\author[hangzhou]{Qin Fang}\ead{12432069@zju.edu.cn},    
\author[SanDiego]{Mamadou Diagne}\ead{mdiagne@ucsd.edu},
\author[hangzhou,ningbo]{Yang Zhu\thanksref{footnoteinfo}}\ead{ zhuyang88@zju.edu.cn}      

\address[hangzhou]{College of Control Science and Engineering, State Key Laboratory of Industrial Control Technology, Zhejiang University, Hangzhou, China}  
\address[SanDiego]{Department of Mechanical and Aerospace Engineering, University of California, San Diego, La Jolla, CA 92093, USA}
\address[ningbo]{Ningbo Innovation Center, Zhejiang University, Ningbo, China}

\begin{keyword}                           
Heterogeneous multi-agent systems; communication delays; predictor design.               
\end{keyword}                             

\begin{abstract}                          
  This paper develops an information outflow predictor-feedback framework for the \emph{exact compensation} of constant but nonuniform communication delays in the output synchronization problem of discrete-time heterogeneous multi-agent systems (MASs). The delays considered here act exclusively on communication channels between neighboring agents, rather than on the plant state, input or output, and thus fall outside the scope of classical predictor-feedback formulations. We show that standard approaches based on time inversion and variational construction of predictor states are insufficient, since even for simple directed acyclic communication graphs the required \emph{prediction horizon} exceeds the delay window. We further illustrate that exact predictor-based compensation is achievable only in the discrete-time setting, while in continuous-time configuration the predictor state becomes non-computable due to the need for a continuum--infinitely many--inaccessible future neighboring-agent information. Motivated by these observations, we propose a distributed prediction architecture in which a prediction-based distributed observer  exactly reconstructs the flow of information exchanged between agents across the network. In addition, we quantify the maximum communication load resulting from information  transmission over the network and show that it worst-case scales linearly with the exosystem dimension, the number of communication edges, the maximum communication delay, and the graph depth. Furthermore, the layer-by-layer information flow eliminates the effect of communication delays after a \emph{finite number of steps}.   Based on this structure, we design prediction-based distributed state-feedback and dynamic output-feedback controllers by combining standard feedback designs with prediction-based distributed observers serving as feedforward components, such that the output of each agent asymptotically tracks the trajectory generated by the exosystem, thereby ensuring output synchronization. The effectiveness of the proposed approach is illustrated on a nonlinear Susceptible--Infected--Recovered (SIR) epidemic model linearized via the Koopman operator framework, where the proposed strategy reduces the infection peak by 200,000 individuals in a population of four million.
\end{abstract}

\end{frontmatter}

\section{Introduction}

\subsection{On the exact compensation of time delays}

\textcolor{black}{ \textbf{From the Smith Predictor to Artstein Reduction.} The development of predictor feedback controllers originated with the pioneering work of Otto J. M. Smith in the 1950s. In the seminal contributions~\cite{smith1957closer,smith1959controller}, the Smith predictor introduced a \emph{delay-free} model to predict the future measured output of a system with input delay, corresponding to the response that would be obtained in the absence of the delay. By decomposing the plant into a delay-free  and a pure delay dynamics, this internal-model structure enables an \emph{exact} cancellation of the delay in the feedback loop, provided that both the plant dynamics and the constant delay are known exactly. Although the original Smith predictor did not explicitly rely on the ideas of \emph{time inversion} and the \emph{variation-of-constants formula}, these concepts emerged in the 1980s as predictor feedback evolved toward \emph{state-space} and \emph{infinite-dimensional} formulations, notably through the Artstein reduction~\cite{artstein1982linear}. The key insight of \cite{artstein1982linear} is that a predictor state can be constructed via the variation-of-constants formula, yielding an equivalent delay-free representation. This construction achieves exact delay compensation and reveals an inherently infinite-dimensional structure, since the dynamics depend on the full input history $u(s)$, $s\in[t-\tau,t]$, where $\tau$ denotes the input delay. This past input history is used to predict the state over the future interval $[t,t+\tau]$, known as the \emph{prediction horizon.}}

 \textcolor{black}{\textbf{PDE interpretation of delays.} This infinite-dimensional viewpoint was further developed in the 1990s and 2000s when it was recognized that input delays can be reformulated as transport partial differential equations, thereby connecting continuous-time predictor feedback with boundary control of hyperbolic equations. In particular, through the framework of PDE backstepping pioneered by  Krsti\'c \cite{krstic2008compensating}, the design of predictor feedback for exact delay compensation has led to a broad range of methodological developments, including extensions to time-varying and state-dependent delays \cite{bekiaris2012compensation,bekiaris2013nonlinear}, delay-adaptive control \cite{krstic2009delay,zhu2020delay}, numerical approximation and implementation of predictor feedback \cite{karafyllis2017predictor}, and predictor-based extremum seeking control under delays \cite{oliveira2022extremum}. Representative recent developments in this direction include unbiased and exponentially convergent ESC under delay \cite{uESC-PDE}, within the extensive body of literature associated with the above monographs (e.g., \cite{zhu2020predictor,zhu2020observer,fang2023prediction,zhang2023adaptive}), with applications ranging from spark-ignition engines and DC motors to Baxter manipulators \cite{bresch2010adaptive,lechappe2016delay,bagheri2019feedback}. To extend these methods to discrete-time settings, various discrete-time predictor frameworks have also been developed \cite{gonzalez2012predictor,gonzalez2013robustness,karafyllis2013robust,choi2016compensation}.}

 \textcolor{black}{Although predictor feedback and related delay compensation techniques achieve exact compensation of input, state, and output delays for both linear and nonlinear continuous-time systems, this property typically holds only under idealized assumptions, with performance degradation arising in the presence of numerical approximation or adaptive implementations. Moreover, most existing results do not explicitly account for \emph{inter-agent communication delays} that naturally arise in cooperative control of MASs over wireless communication channels. Unlike input or output delays associated with an individual plant, such delays affect the exchange of neighboring-agent information and relative-state measurements in distributed control protocols, making the closed-loop dynamics depend on delayed and topology-coupled information. This feature reduces the direct applicability of standard predictor-based designs in modern cyber-physical and networked environments. \emph{In this context, the present work aims to extend predictor feedback methodologies to MASs with communication delays, developing a framework that preserves the exact compensation principle.}}

\subsection{Communication delays in heterogeneous multi-agent systems}

\noindent\textbf{Synchronization versus consensus.} Distributed cooperative control of multi-agent systems (MASs) has attracted significant research attention over the past decades, driven by diverse applications such as unmanned aerial vehicles (UAVs) formation control \cite{zhang2022distributed}, networked robotic systems \cite{wang2018fixed} and smart grids \cite{cai2015distributed}. 
The fundamental objective is to ensure that the group as whole reaches an agreement on certain quantities of interests through a distributed protocol relying solely on locally exchanged information over a communication network--commonly referred to as consensus or synchronization. \textcolor{black}{There is a subtle difference between these two notions (see \cite{scardovi2009synchronization, wieland2011internal,FB-LNS} and references therein). Roughly speaking, consensus primarily concerns how communication constraints and graph connectivity affect the convergence of multiple agents to a common value, often under relatively simple individual dynamics. In contrast, synchronization is more concerned with coordinating systems with nontrivial intrinsic dynamics, and seeks to ensure that all agents asymptotically evolve along a common trajectory through the exchange of relative information.} Nowadays, numerous distributed control algorithms have been developed to achieve consensus or synchronization under a variety of communication topologies, system dynamics, and performance requirements. In the following, we discuss the synchronization problem that is the focus of the present work.
 
\noindent\textbf{Delay-free synchronization.} Generally, MASs can be classified into two categories: homogeneous and heterogeneous MASs. Early research primarily focused on homogeneous MASs, in which all agents are assumed to possess identical dynamics, ranging from simple first-order models \cite{jadbabaie2003coordination,tanner2003stability} to more general high-order dynamics \cite{wieland2008consensus,ren2006high,scardovi2009synchronization,xu2018consensusability}. In contrast to homogeneous MASs, heterogeneous MASs involve agents with non-identical, possibly different-order dynamics, making them more suitable for modeling real-world systems. In such cases, output synchronization often becomes the primary goal. There are two methods to deal with output synchronization problems of heterogeneous MASs: distributed internal model principle \cite{wieland2009internal, wieland2011internal,zuo2017output} and distributed feedforward approach \cite{su2011cooperative, su2012cooperative, cai2015leader, cai2017adaptive, huang2016cooperative}. In contrast to the distributed internal model principle, the distributed feedforward approach does not require satisfying the transmission zero condition---a condition that is never met when the output dimension of the system exceeds its input dimension.
The distributed feedforward approach, first proposed in \cite{su2011cooperative}, employs a distributed observer for each agent to estimate the state of the exosystem and design distributed dynamic controllers, thereby addressing the cooperative output regulation problem of continuous-time heterogeneous linear MASs. Subsequently, the approach was applied to deal with different types of problems, such as discrete-time systems \cite{huang2016cooperative,zhang2024mean}, switching communication network \cite{su2012cooperative,lu2016distributed,liu2018leader}, uncertain exosystem matrix \cite{cai2015leader, cai2017adaptive} and communication channels with packet loss \cite{zhang2024mean}. Nevertheless, numerous challenges remain to be addressed.

\noindent \textbf{ Synchronization under communication delays.} For MAS coordination over wireless communication networks, communication channels among agents are frequently impaired by communication delays, as information exchange is constrained by physical transmission, processing latencies, and other network-induced effects. Moreover, time delays frequently induce instability (see \cite{fridman2014introduction, Jean-Pierre2003Time-delay}), posing challenges to the stability and performance of MAS coordination. Existing work on synchronization problem with communication delays has primarily focused on continuous-time systems (see \cite{moreau2004stability, olfati2004consensus,zhu2010leader, zhou2014consensus,lu2016distributed} and references therein).  Output synchronization problem of  discrete-time  heterogeneous MASs with communication delays have been addressed in \cite{yin2013event, xu2017consensus,liu2022scale, zhang2024fully}.
 In particular, in \cite{xu2017consensus}, A modified distributed observer was proposed as a basis for distributed controller design in MASs subject to arbitrarily bounded, nonuniform, time-varying communication delays. Nevertheless, the consensus analysis in \cite{xu2017consensus} relies on the consensus of an equivalent delayed MAS, indicating that the delay is merely tolerated rather than exactly compensated.

\noindent \textbf{Synchronization under predictor feedback design.} The predictor feedback framework has been extended to distributed cooperative control of MASs with delays, although most existing results are restricted to networks with homogeneous dynamics. In particular, nested predictor-based feedback protocols were proposed in \cite{liu2018consensus} for discrete-time homogeneous MASs with state, input, and communication delays, an observer–predictor approach was developed in \cite{liu2020observer}, and a discrete-time predictor for continuous-time homogeneous MASs was obtained via discretization in \cite{ponomarev2017discrete}.
\textcolor{black}{Extending predictor-based control to heterogeneous MASs is more challenging. A notable contribution is \cite{tan2013consensus}, which proposed a prediction-based distributed control scheme for discrete-time heterogeneous MASs with uniform delays under the networked predictive control system (NPCS) framework, using a Luenberger-inspired one-step-ahead predictor to construct multi-step predictions. However, exact communication-delay compensation for heterogeneous MASs remains largely unresolved. In \cite{tan2013consensus}, the prediction error vanishes only asymptotically, leading to inexact compensation at finite time, and the approach further requires agents to have identical state dimensions, which limits broad applicability to output synchronization settings. Besides, the method in \cite{tan2013consensus} requires each local agent to know the system matrix $A_j$, input matrix $B_j$, and output matrix $C_j$ of its neighboring agents. In other words, this implies that each local agent has some prior global information. }
\subsection{Contributions}

\textcolor{black}{In this work, we address the output synchronization problem for discrete-time heterogeneous MASs subject to constant and distinct communication delays via predictor feedback design. The motivation for this problem stems primarily from meta-population SIR models \cite{zhong2021country} and V2V-based connected and automated vehicles (CAVs) \cite{zhu2020v2v}, where communication delays naturally arise from inter-region mobility and wireless communication, respectively. The main contributions are summarized as follows:}

\begin{itemize}[leftmargin=*, itemsep=0pt]
\item 
\textbf{Limitations of exact outflow prediction in directed acyclic MASs with communication delays.} Although the communication graph is assumed to be directed and acyclic, i.e., the simplest case, we show that exact prediction of information outflow in MASs is not achievable using classical predictor feedback designs based on time inversion and the variation-of-constants formula, for both homogeneous and heterogeneous systems in continuous and discrete time. The continuous-time predictor is not impractical merely because it involves future information. The fundamental difficulty is that exact delay compensation in the continuous-time setting requires access to signal values over an entire future interval, i.e., a continuum of future-time instants rather than a finite set of points. Furthermore, in both continuous- and discrete-time settings, communication delays induce an inherent information-transmission mismatch: agents can only access delayed neighbor information, whereas predictor construction requires an extended prediction horizon beyond  $[t, t+\tau].$ This structural limitation prevents the construction of the classical predictor state. To address the second issue, one may require neighboring agents to transmit additional future information, which motivates distributed predictors in the discrete-time setting, where such information can be generated recursively in a step-by-step manner.

    \item \textbf{Information outflow prediction design.}
    We develop an exact information outflow delay compensation tailored for communication delays and requiring the definition of an \emph{extended prediction horizon.} Inspired by the classical state predictor design (see \cite{artstein1982linear, krstic2009delay}), we design a discrete-time predictor framework that transmits future information over a directed acyclic graph, thereby enabling the exact compensation of distinct communication delays in an agent-by-agent manner. In addition, we introduce a distributed predictor that serves as an auxiliary mechanism to transmit additional future information, thus facilitating the predictor design of neighboring agents.
    \item  \textbf{Prediction-based distributed control under finite-time-step delay compensation.} We design a prediction-based distributed observer to estimate the state of the exosystem. By integrating the standard distributed observer design with the proposed predictor and distributed predictors, the modified distributed observers can achieve consensus despite the presence of communication delays.  Then,  we propose two prediction-based distributed control laws, state feedback and dynamical output feedback, to guarantee the output synchronization of the heterogeneous MASs. In contrast to \cite{xu2017consensus,tan2013consensus}, the proposed approach eliminates the effect of delays within a finite number of time steps and consequently achieves better transient performance. 

    \item \textbf{Koopman-based validation on epidemic spread under communication delays.} Finally, we apply the proposed methods to a susceptible–infected–recovered (SIR) epidemic model with communication delays. Using Koopman operator theory, we construct a linear Koopman-based approximation to embed the system within the proposed framework. Simulation results show that the proposed delay-compensation strategy effectively reduces the infection peak compared with the case without delay compensation.
\end{itemize} 

The paper is organized as follows. Section \ref{problem-statement} states the problem and motivates the discrete-time setting. Section \ref{distributed-obs}
presents the  prediction-based distributed observer design. Section \ref{adsec} established the exact compensation of communication delays and  states the realization of distributed observers consensus. Section \ref{distributed-control} provides two distributed control designs to achieve the output synchronization under communication delays. Section \ref{simulation} presents numerical simulations and an application to an epidemic model linearized via Koopman operator theory. The paper ends with the concluding remarks and  perspectives in Section \ref{conclusion}.

\noindent {\bf Notation.} 
Let $\mathbb{N}$, $\mathbb{N}^+$, $\mathbb{R}$, $\mathbb{R}^+$, $\mathbb{R}^n$, and $\mathbb{R}^{n\times m}$ denote the set of all natural numbers, all positive integers, real numbers, positive real numbers, real column vectors of dimension $n$, and real matrices of dimension $n\times m$, respectively.  
Let $N$, $r$, $p$, $q$, $m_i$ and  $n_i\in\mathbb{N}^+$.
For a vector $X \in \mathbb{R}^n$, its Euclidean $2$-norm is denoted as $|X|$. 
For a matrix $A = \{a_{ij}\}\in \mathbb{R}^{n\times n}$, $\rho(A)$ and $A^\mathrm{T}$ represent its spectral radius  and transpose, respectively. 
We indicate its induced spectral norm as $\|A\| = (\lambda_{\max}(A^\mathrm{T}A))^{\frac{1}{2}}$. 
$\mathbf{I}_n$ denote the $n$-by-$n$ identity matrix. 
$\mathbf{0}_n$ and $\mathbf{1}_n$ represent $n$-dimensional column vectors with all entries equal to zero and one, respectively. $\mathbf{0}_{n \times m }$ represents $n$-by-$m$ matrix with all entries equal to zero.
$\text{diag}[a_1,a_2,\ldots,a_n]$ denotes a $n$-by-$n$ diagonal matrix with $a_1,a_2,\ldots,a_n$  as its diagonal entries.


\noindent \textbf{Graph theory and definitions.} In a weighted graph $\mathcal{G} = (\mathcal{N}, \mathcal{E}, \mathcal{A})$, the set of nodes is given by $\mathcal{N} = \{0,1, 2, \ldots, N\}$, and the set of edges is $\mathcal{E} \subseteq \mathcal{N} \times \mathcal{N}$. The matrix $\mathcal{A} = \{a_{i,j}\} \in \mathbb{R}^{(N + 1) \times (N+1)}$ is the weighted adjacency matrix, where $a_{i,j} \ne 0$ if $(j,i) \in \mathcal{E}$, and $a_{i,j} = 0$ otherwise. An edge $(j,i)$ indicates that \textbf{Agent} $\boldsymbol{i}$ receives information from \textbf{Agent} $\boldsymbol{j,}$ but not necessarily vice versa.
An ordered sequence of nodes $(v_1, v_2, \ldots, v_m)$ with $m \ge 1$ is called a path $\pi$ if $(v_i, v_{i+1}) \in \mathcal{E}$ for all $i \in {1, \ldots, m-1}$. The length of a path $\pi$ consisting of $m$ connected edges is defined as $\mathbb{L}(\pi) = m$. The weight of a path $\pi$, denoted by $w(\pi)$, is defined as the sum of the weights of its edges. For $i < j$, let $\Pi_{i,j}$ denote the set of all paths $\pi$ from node $i$ to node $j$ such that $w(\pi) < \infty$. 
\section{Problem statement and motivation}\label{problem-statement}

\subsection{System's dynamics  and preliminaries}
\textcolor{black}{Consider a discrete-time heterogeneous linear MAS consisting of $N$ agents, where the dynamic model of \textbf{Agent} $\boldsymbol{i,}$ $i\in\{1,2,\ldots,N\}$, is given by}:
\begin{align}
	\begin{array}{l}
		x_i(k+1) = A_i x_i(k) + B_i u_i(k),\\
		y_i(k) = C_{i}x_i(k),\ i=1,2,\ldots,N, 
	\end{array} \label{equ:plant}
\end{align}
where $x_i(k)\in\mathbb{R}^{n_i}$ denotes the state of \textbf{Agent} $\boldsymbol{i,}$ $u_i(k)\in\mathbb{R}^{m_i}$ represents the control input of \textbf{Agent} $\boldsymbol{i,}$  $y_i(k)\in\mathbb{R}^{p}$ is measured output of \textbf{Agent} $\boldsymbol{i}$. In addition, $A_i\in\mathbb{R}^{n_i \times n_i}$, $B_i\in\mathbb{R}^{n_i \times m_i}$ and  $C_i\in\mathbb{R}^{p \times n_i}$ represent the system, input and output matrices of \textbf{Agent} $\boldsymbol{i,}$ respectively. 

Additionally, we introduce the following exosystem, which serves as the leader and generates the exogenous signal to be tracked by agents:
\begin{align}
	\upsilon(k+1) = S\upsilon(k), \label{equ:exosystem}
\end{align}
where $\upsilon(k) \in \mathbb{R}^q$ is the state of the exosystem. 
\begin{figure}
	\centering
	\begin{tikzpicture}[scale=0.7]
		
		\draw[fill=Maroon!50,very thick] (0,0.1) circle (0.6) node {\scriptsize $0$};
		
		\draw[->, dashed, Red, very thick, >=latex]  (0,-0.5) ->  (0,-1.25) ;
		
		\draw[->, dashed, Red, very thick, >=latex]  (-0.6,0.1) --  (-4.5,0.1) --  (-4.5,-4.25) --  (-3.75,-4.25);
		
		\draw[->, dashed, Red, very thick, >=latex]  (-0.6,0.1) --  (-4.5,0.1) --  (-4.5,-6.5) --  (-3.75,-6.5);
		
		\draw[->, dashed, Red, very thick, >=latex]  (-0.6,0.1) --  (-4.5,0.1) --  (-4.5,-8.75) --  (-3.75,-8.75);
		
		\draw[fill=NavyBlue!50,very thick] (-2.75,-2) circle (0.6) node {\scriptsize $1,1$};
		\draw[fill=NavyBlue!50,very thick] (-1.25,-2) circle (0.6) node {\scriptsize $1,2$};
		\draw[fill=NavyBlue!50,very thick] (1.25,-2) circle (0.6) node {\scriptsize $1,N_1$};
			\node at (0,-2) {...};
		\begin{pgfonlayer}{background}
			\draw[rounded corners,very thick] (-3.75,-1.25) rectangle (3.75,-2.75) node[midway, xshift=2.1cm] {\scriptsize Layer $1$}; 
		\end{pgfonlayer} 
		
		\draw[->, dashed, Blue, very thick, >=latex]  (0,-2.75) ->  (0,-3.5) ;
		
		\draw[->, dashed, Blue, very thick, >=latex]  (3.75,-2) --  (4.5,-2) --  (4.5,-6.5) --  (3.75,-6.5);
		
		\draw[->, dashed, Blue, very thick, >=latex]  (3.75,-2) --  (4.5,-2) --  (4.5,-9) --  (3.75,-9);
		
		\draw[fill=NavyBlue!50,very thick] (-2.75,-4.25) circle (0.6) node {\scriptsize$2,1$};
		\draw[fill=NavyBlue!50,very thick] (-1.25,-4.25) circle (0.6) node {\scriptsize$2,2$};
		\draw[fill=NavyBlue!50,very thick] (1.25,-4.25) circle (0.6) node {\scriptsize $2,N_2$};
		\node at (0,-4.25) {...};
		\begin{pgfonlayer}{background}
			\draw[rounded corners,very thick] (-3.75,-3.5) rectangle (3.75,-5) node[midway, xshift=2.1cm] {\scriptsize Layer $2$}; 
		\end{pgfonlayer} 
		
		\draw[->, dashed, purple, very thick, >=latex]  (0,-5) ->  (0,-5.75) node[midway, right]{...};
		
		\draw[->, dashed, purple, very thick, >=latex]  (3.75,-4.25) --  (4.125,-4.25) --  (4.125,-8.5) --  (3.75,-8.5);
		
		\draw[fill=NavyBlue!50,very thick] (-2.75,-6.5) circle (0.6) node {\scriptsize$i,1$};
		\draw[fill=NavyBlue!50,very thick] (-1.25,-6.5) circle (0.6) node {\scriptsize$i,2$};
		\draw[fill=NavyBlue!50, very thick] (1.25,-6.5) circle (0.6) node {\scriptsize $i,N_i$};
		\node at (0,-6.5) {...};
		\begin{pgfonlayer}{background}
			\draw[rounded corners,very thick] (-3.75,-5.75) rectangle (3.75,-7.25) node[midway, xshift=2.1cm] {\scriptsize Layer $i$}; 
		\end{pgfonlayer} 
		
		\draw[->, dashed, darkgray, very thick, >=latex]  (0,-7.25) ->  (0,-8);
		
		\draw[fill=NavyBlue!50,very thick] (-2.75,-8.75) circle (0.6) node {\scriptsize$M,1$};
		\draw[fill=NavyBlue!50,very thick] (-1.25,-8.75) circle (0.6) node {\scriptsize$M,2$};
		\draw[fill=NavyBlue!50,very thick] (1.25,-8.75) circle (0.6) node {\scriptsize $M,N_M$};
		\node at (0,-8.75) {...};
		\begin{pgfonlayer}{background}
			\draw[rounded corners,very thick] (-3.75,-8) rectangle (3.75,-9.5) node[midway, xshift=2.1cm] {\scriptsize Layer $M$}; 
		\end{pgfonlayer} 
	\end{tikzpicture}
	\caption{Multi-agent grouping structure.}
	\label{fig_structure}
\end{figure}
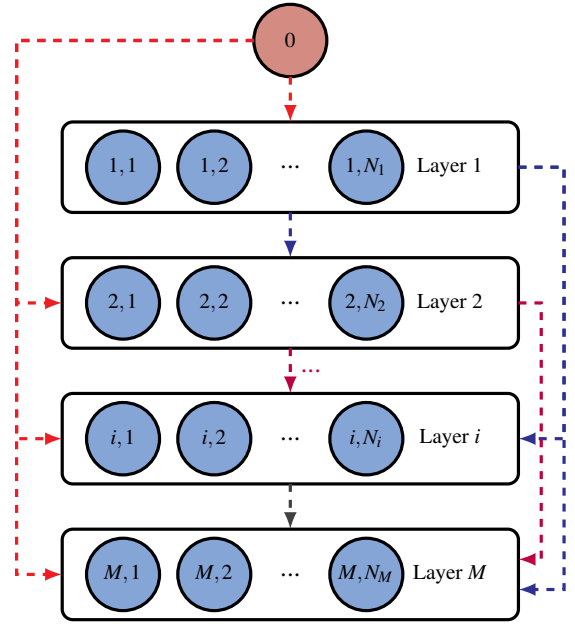
\begin{assumption} \label{assum1}
	The graph $\mathcal{G}$ is a directed acyclic graph (DAG).
\end{assumption}
\begin{assumption} \label{assum3}
	The pair $(A_i,B_i)$ is stabilizable, and the pair $(C_{i}, A_i)$ is detectable for $i = 1, 2, \ldots, N$.
\end{assumption}
\begin{assumption} \label{assum4}
	For $i=1,2,\ldots,N$, the following regulator equations:
	\begin{align}
		\begin{array}{rl}
			X_iS =& A_iX_i + B_i U_i,\\
			0 =& C_i X_i + F,
		\end{array}\label{equ:regulation}
	\end{align}
	have solution pairs $(X_i,U_i)$, where $X_i\in\mathbb{R}^{n_i \times q}$, $U_i\in\mathbb{R}^{m_i \times q}$ and $F\in\mathbb{R}^{p\times q}$. \textcolor{black}{In particular, $F$ denotes the reference signal generation matrix, which maps $\upsilon$ to the reference output.}
\end{assumption}

\textcolor{black}{To make the direction of information transmission explicit, we partition the graph into $M+1$ layers, where layer $0$ consists of the exosystem (or leader), and each of the remaining layers contains $N_i$ nodes, $i=1,2,\ldots,M$, with $\sum_{i=1}^M N_i=N$, as illustrated in Fig.~\ref{fig_structure}. Under Assumption \ref{assum1},  such a layered representation can always be constructed according to a topological ordering of the nodes. In this representation, the layers are arranged such that information is transmitted unidirectionally from an upper layer to lower layers. Here, for the convenience of analysis, we define $M=N$, thereby each layer assigns a single agent. } 

Furthermore, the weighted adjacency matrix $\mathcal{A}$ and Laplace matrix $\mathcal{L}$ present the following lower-triangular forms:
\begin{gather*}
	\mathcal{A} =\begin{bNiceMatrix}[small,xdots/shorten=6pt]
		0 & 0 & \cdots & 0 & 0\\
		a_{1,0} & 0 & \cdots & 0 & 0\\
		\Vdots & \Vdots & \Ddots & \Vdots & \Vdots \\
		a_{N-1,0} & a_{N-1,1} & \cdots & 0 & 0 \\
		a_{N,0} & a_{N,1} & \cdots & a_{N,N-1} & 0 \\
	\end{bNiceMatrix}\in\mathbb{R}^{(N+1)\times (N+1)},\\
	\mathcal{L} =\begin{bNiceArray}{c|cw{c}{1cm}c}[small,margin]
		0 & 0 & \Cdots & 0\\
		\hline
		-a_{1,0} & \Block{4-3}{\mathcal{H}+\mathcal{D}_0} &  &  \\
		\Vdots & & &  \\
		-a_{N-1,0} &   &    &   \\
		-a_{N,0} &   &    &   
	\end{bNiceArray}\in\mathbb{R}^{(N+1)\times (N+1)},
\end{gather*}
where $\mathcal{D}_0 = \text{diag}[a_{1,0},a_{2,0},\ldots,a_{N,0}]\in\mathbb{R}^{N\times N}$,
\begin{gather*}
	\mathcal{H} =  \begin{bNiceMatrix}[small,xdots/shorten=6pt]
		0 & 0& \cdots & 0 & 0\\
		-a_{2,1} & a_{2,1} & \cdots & 0 & 0\\
		\Vdots & \Vdots & \Ddots & \Vdots & \Vdots \\
		-a_{N-1,1} & -a_{N-1,2} & \cdots & \sum_{j=1}^{N-2}a_{N-1,j} & 0 \\
		-a_{N,1} & -a_{N,2} & \cdots & -a_{N,N-1} & \sum_{j=1}^{N-1}a_{N,j} \\
	\end{bNiceMatrix}\in\mathbb{R}^{N\times N}.
\end{gather*}
In addition, in a distributed control setting, not all agents are necessarily required to have direct access to the exosystem information. This is reflected by the condition that  \textcolor{black}{$a_{1,0} > 0$ and $a_{i,0} \ge 0$, $i=2,3,\ldots,N$.}  Moreover, the fact that each layer contains a single agent implies that $a_{i,i-1}>0$, $i=2,3,\ldots,N$.
\begin{remark}
	\textcolor{black}{The design constraints and information structure considered in this paper are summarized as follows.
		\begin{itemize}[leftmargin=*, itemsep=0pt]
			\item The communication graph and the distinct communication delays  are fixed and known. The directed acyclic graph is  used to determine the layer structure.
			\item The agent dynamics $(A_i,B_i,C_i)$ are local information and are not required to be shared among agents. Each agent only requires its own system matrices to design the corresponding local feedback and feedforward gains, based on the solution of its own regulator equations.
			\item During online implementation, each agent only uses its own state, its distributed observer state, and  the prediction of the delayed information received from its neighbors through the communication channel.
		\end{itemize}
		Thus, the proposed method is distributed in implementation.} 
\end{remark}

 In the sequel, we aim to design prediction-based distributed control strategies to exactly compensate for these communication delays  in finite time steps. Since the agents may exhibit heterogeneous dynamics, the dimensions of their states may differ. \textcolor{black}{Under this setting, our objective is to achieve output synchronization. To this end, our design consists of two main components:
 \begin{enumerate} [leftmargin=*, itemsep=0pt]
 \item The construction of a \emph{distributed observer} $\xi_i(k)$ that achieves  consensus, in Lemma \ref{lemma2}, with respect to the exogenous signal $\upsilon(k)$ generated by \eqref{equ:exosystem}, namely, $\lim_{k\to\infty}|\xi_i(k) - \upsilon(k)|=  0$, $i=1,2,\ldots,N$. 
 \item  Achieving the output synchronization of the heterogeneous MAS, i.e., $\lim_{k\to\infty}|y_i(k)-y_j(k)|=0$, for any $i,j\in\{1,2,\ldots,N\}$, in Theorems \ref{theorem1} and \ref{theorem2} 
 \end{enumerate} 
 The result of the first stage serves as a key basis for establishing the output synchronization result in the second stage.}


\subsection{Motivation for a Discrete-Time Framework}\label {discretetime-motivation}
In this section, we show how the structure of the problem prevents the design of a continuous-time predictor to exactly compensate for the delays. We emphasize that the design of a predictor is in sharp contrast with the results  in \cite{su2011cooperative}, where a distributed feedforward approach was employed to realize  output synchronization of MASs in continuous-time under delay-free case. To simplify the exposition and illustrate our point, we consider only identical delays in the sequel.

Consider the following continuous-time counterpart:
	\begin{align}
		\begin{array}{l}
			\dot{x}_i(t) = A_i x_i(t) + B_i u_i(t),\\
			y_i(t) = C_{i}x_i(t),\ i=1,2,\ldots,N, 
		\end{array} \label{equ:plant_continuous}
	\end{align}
	along with the continuous-time exosystem
	\begin{align}
		\dot{\upsilon}(t) = S\upsilon(t),
	\end{align}
	where the representations and dimensions of the notations are identical to those in \eqref{equ:plant}--\eqref{equ:exosystem}.

 In \cite{su2011cooperative},  the following distributed observer is proposed, whose state asymptotically converges to $\upsilon$:
\begin{align*}	
	\dot{\xi}_i(t) =&\ \textstyle S\xi_i(t) -  \beta a_{i,0} (\xi_i(t) - \upsilon(t)) \notag\\
	&  \textstyle - \beta \sum_{j=1}^{i-1}a_{i,j} (\xi_i(t) - \xi_{j}(t)),
\end{align*}
where $\xi_i(t) \in \mathbb{R}^p$ represents the state of the distributed observer for each \textbf{Agent} $\boldsymbol{i}$ can be derived. In sharp contrast with \cite{su2011cooperative}, we consider the presence of  uniform communication delays, all equal to $\tau \in \mathbb{R}^+,$ which prevent receivers from instantaneously obtaining information from their senders. Therefore, we seek to compensate for the effects of delayed information through prediction. To this end, the following modified distributed observers are constructed:
\begin{align*}	
	\dot{\xi}_i(t) =&\ \textstyle S\xi_i(t) -  \beta a_{i,0} (\xi_i(t) - P_0(t - \tau)) \notag\\
	&  \textstyle - \beta \sum_{j=1}^{i-1}a_{i,j} (\xi_i(t) - P_{j}(t-\tau)),
\end{align*}
where $P_j(t)$ is the prediction of the  distributed observer state $\xi_j$, $j=0,2,\ldots,i-1$, obtained via the classical time-inversion technique and the variation-of-constants formula,  $P_j(t)=\xi_j(t+\tau).$ Following \cite{krstic2009delay}, the predictors for the exosystem and each agent  are given by
\begin{align*}
    P_0(t)  = \upsilon(t+\tau ) =&\ e^{S\tau}\upsilon(t),\\
	P_i(t) = \xi_i(t+\tau ) =&\ \textstyle e^{\mathbf{S}_i \tau} \xi_i(t) + \beta a_{i,0} \int_{t - \tau}^{t}e^{\mathbf{S}_i(t-s)}P_0(s)ds \notag\\
	& \textstyle    +  \beta \sum_{j=1}^{i-1} a_{i,j} \int_{t - \tau}^{t}e^{\mathbf{S}_i(t-s)} P_{j}(s)ds,
\end{align*}
where $\mathbf{S}_i = S - \beta \sum_{j=0}^{i-1}a_{i,j}\mathbf{I}_p$.

\noindent \textcolor{black}{\textbf{Infeasibility argument.} The above predictors are not implementable since at $t=t^*$,  \textbf{Agent} $\boldsymbol{i}$, only has access to the information from its neighbors at $t = t^*-\tau$, namely, $P_j(t^*-\tau)$,  whereas the predictor requires  knowledge of the signals over the entire interval $[t^*-\tau,t^*]$ through the integral operator induced by the variation-of-constant-formula. In other words, the prediction task is \emph{noncausal:} predicting future information requires access to information that is not yet available to \textbf{Agent} $\boldsymbol{i}$.  Such a \emph{causality and information-availability} issue is not unique to the continuous-time case; it also appears in the discrete-time case if a standard predictor is implemented directly, which would be discussed later. Nevertheless,  \emph{an information outflow predictor} obtained through recursive formula leads to a solution in discrete-time setting. Such extended predictors rely on an algorithm that requires  neighboring agents  to transmit additional future information at $t=t^*-\tau$ via distributed predictors.}

\textcolor{black}{Nevertheless, in continuous-time setting, the design of such  extended predictors gives rise to  a new impediment  that is  illustrated in the following argument. Let us specialize $P_i$ to $i=2$ to obtain the following predictor state 
	\textcolor{black}{\begin{align*}
		P_2(t) =&\ \textstyle e^{\mathbf{S}_2 \tau} \xi_2(t) + \beta a_{2,0} \int_{t - \tau}^{t}e^{\mathbf{S}_2(t-s)}P_0(s)ds \notag\\
		& \textstyle    +  \beta a_{2,1} \int_{t - \tau}^{t}e^{\mathbf{S}_2(t-s)} P_{1}(s)ds,
	\end{align*}}
assuming that \textbf{Agent} $\boldsymbol{1,}$ at time $t = t^* - \tau$ transmits the information $\bar{P}_1(t) = \int_{t}^{t+\tau}e^{\mathbf{S}_2(t+\tau-s)}P_1(s)ds$ directly to its neighbor \textbf{Agent} $\boldsymbol{2,}$ who can thereby receive $\bar{P}_1(t^*-\tau)$ at $t = t^*$.  It becomes obvious that a \emph{key obstacle} to practically realize an \emph{exact} continuous-time predictor through the computation of $\bar{P}_1(t)$ is the \emph{absolute necessity} to predict the future signals in the entire continuum of points in the interval $[t, t+\tau]$, i.e., infinitely many values.} 

\noindent \textbf{Feasibility in discrete-time setting.} Compared with the continuous-time setting, where future information is obtained via an \emph{integral operator,} the information outflow predictor design for discrete-time systems can be implemented through iteration, allowing future information to be obtained step by step in a recursive manner. The specific implementation details are provided in the design of the predictors and distributed predictors presented in the sequel.

\section{Distributed observer and predictor design}\label{distributed-obs}
In this section, we introduce  a prediction-based distributed observer design for discrete-time heterogeneous MASs.

In the following, $\tau_{i,j}\in\mathbb{N}^+$ denotes the communication delay from \textbf{Agent~$\boldsymbol{i}$} to \textbf{Agent~$\boldsymbol{j.}$} Next, we redefine the usual notion of the longest path from node $i$ to node $j$ as follows:
\begin{definition}[Extended prediction horizon]\label{def1}
     Let the longest path from  node $i$ to node $j$ be defined as
    \begin{align}
    	\textstyle \pi_{i,j} = \arg \max_{\pi \in \Pi_{i,j}}\{w(\pi) - \mathbb{L}(\pi)\},
    \end{align}
    and let the corresponding weight $w_{i,j},$ be defined as
    \begin{align}
    	w_{i,j} = \textstyle \max_{\pi \in \Pi_{i,j}}\{w(\pi) - \mathbb{L}(\pi)\}, \label{equ: longest_path}
    \end{align}
    where the standard edge weights, $a_{i,j}$, are replaced by the corresponding communication delay lengths, $\tau_{i,j}$. Namely, the weight of a path $\pi=(v_1,v_2,\ldots,v_m)$ is defined as 
     \begin{align}
     	w(\pi) =\textstyle  \sum_{i=1}^{m-1}\tau_{v_i,v_{i+1}}.
     \end{align}
    Then, the  extended prediction horizon  of  information transmitted from \textbf{Agent~$\boldsymbol{i}$} to \textbf{Agent~$\boldsymbol{j}$} is denoted by
    \begin{align}
    	\mathbb{P}_{i,j} = \{0,1, \ldots, w_{j,N}\},\ j =1,2,\ldots, N,\ i<j. \label{equ:prediction_horizon_heter}
    \end{align}
\end{definition}
\subsection{Infeasibility of standard predictor design and extended prediction horizon in discrete-time}\label{infeas}
\textcolor{black}{We begin with a three-agent network together with the exosystem. The purpose of this example is to show why a standard predictor is not directly implementable under communication delays, and why additional future information must be generated and transmitted. }
\begin{figure}[!h]
	\centering
	\begin{tikzpicture}[scale=0.4]
		
		\draw[fill=Maroon!50,very thick] (0,0) circle (0.5) node {\scriptsize $0$};
		
		\draw[fill=NavyBlue!50,very thick] (2,0) circle (0.5) node {\scriptsize $1$};
		\draw[fill=NavyBlue!50,very thick] (4,0) circle (0.5) node {\scriptsize $2$};
		\draw[fill=NavyBlue!50,very thick] (6,0) circle (0.5) node {\scriptsize $3$};
		
		\draw[->, dashed, Blue, very thick, >=latex]  (0,0.5) to [out=30,in=150] (4,0.5);
		\draw[->, dashed, Blue, very thick, >=latex]  (0,0.5) to [out=30,in=150] (6,0.5);
		\draw[->, dashed, Blue, very thick, >=latex]  (2,-0.5) to [out=-30,in=-150] (6,-0.5);
		
		\draw[->, dashed, Blue, very thick, >=latex]  (0.5,0) ->  (1.5,0) ;
		
		\draw[->, dashed, darkgray, very thick, >=latex]  (2.5,0) ->  (3.5,0);
		
		\draw[->, dashed, darkgray, very thick, >=latex]  (4.5,0) ->  (5.5,0);
	\end{tikzpicture}
	\caption{Communication Topology.}
	\label{fig_example}
\end{figure}

\textcolor{black}{A simple communication topology  shown in Fig. \ref{fig_example} is used as an illustrative case study. Based on the standard distributed observer design, we construct the following prediction-based distributed observer \cite{huang2016cooperative}:
\begin{align}	
	\xi_1(k+1) =&\ \textstyle S\xi_1(k) -  \beta S a_{1,0} (\xi_1(k) - \Upsilon_1(k-\tau_{0,1} )), \label{equ:example1}\\
	\xi_2(k+1) =&\ \textstyle S\xi_2(k) -  \beta S a_{2,0} (\xi_2(k) - \Upsilon_2(k-\tau_{0,2} )) \notag\\
	&  \textstyle - \beta Sa_{2,1} (\xi_2(k) - \Xi_{1,2}(k-\tau_{1,2})), \label{equ:example2}\\
	\xi_3(k+1) =&\ \textstyle S\xi_3(k) -  \beta S a_{3,0} (\xi_3(k) - \Upsilon_3(k-\tau_{0,3} )) \notag\\
	&  \textstyle - \beta S a_{3,1} (\xi_3(k) - \Xi_{1,3}(k-\tau_{1,3}))\notag\\
	&  \textstyle - \beta S a_{3,2} (\xi_3(k) - \Xi_{2,3}(k-\tau_{2,3})), \label{equ:example3}
\end{align}
where $\Upsilon_i(k)$ denotes the exosystem prediction associated with \textbf{Agent} $\boldsymbol{i,}$ and $\Xi_{i,j}(k)$ denotes the prediction generated by \textbf{Agent} $\boldsymbol{i}$ and used by \textbf{Agent} $\boldsymbol{j,}$  and  $\beta$ is the coupling gain.  These predictors are introduced so that each \textbf{sender} can predict its own future information and transmit it to its neighboring \textbf{receivers,} thereby compensating for communication delays.}

\textcolor{black}{Based on \eqref{equ:example1}--\eqref{equ:example2}, through step-by-step iteration, the  $\tau_{i,j}$-step ahead predictions for \textbf{Agent~$\boldsymbol{1}$} and \textbf{Agent~$\boldsymbol{2}$} are given
\begin{align}
	&\ \xi_1(k + \tau_{1,2}) 	= \textstyle 	\hat{S}_1^{\tau_{1,2}}\xi_1(k) + \beta \sum_{j=1}^{\tau_{1,2}}  \hat{S}_1^{j-1} a_{1,0} S \notag\\
    & \textstyle \qquad\qquad\qquad \times  \Upsilon_1(k - \tau_{0,1} + \tau_{1,2} - j ), \label{equ:example_pre_1} \\
	&\ \xi_1(k + \tau_{1,3}) = \textstyle \hat{S}_1^{\tau_{1,3}}\xi_1(k) + \beta \sum_{j=1}^{\tau_{1,3}}  \hat{S}_1^{j-1} a_{1,0} S\notag\\
    & \textstyle \qquad\qquad\qquad \times \Upsilon_1(k - \tau_{0,1} + \tau_{1,3} - j  ), \label{equ:example_pre_2}\\
	&\ \xi_2(k + \tau_{2,3}) = \textstyle \hat{S}_2^{\tau_{2,3}}\xi_2(k) + \beta \sum_{j=1}^{\tau_{2,3}}  \hat{S}_2^{j-1} a_{2,0} S \notag\\
    & \textstyle \qquad\qquad\qquad \times \Upsilon_2(k - \tau_{0,2} + \tau_{2,3} - j ) + \beta \sum_{j=1}^{\tau_{2,3}} \hat{S}_2^{j-1} a_{2,1} S \notag\\
	& \textstyle \qquad\qquad\qquad \times  \Xi_{1,2}(k - \tau_{1,2} + \tau_{2,3} - j ), \label{equ:example_pre_3}
\end{align}
where $\hat{S}_1 = S - \beta Sa_{1,0}$ and $\hat{S}_2 = S - \beta S a_{2,0} - \beta S a_{2,1}$.}

\textcolor{black}{However, \emph{as in the continuous-time case}, such a classical predictor is not directly implementable in the discrete-time setting when communication delays are present.  To see this, suppose that \textbf{Agent~$\boldsymbol{2}$} intends to compute $\xi_2(k^*+\tau_{2,3})$ at time $k=k^*$. According to \eqref{equ:example_pre_3}, this computation requires information from the upper-layer nodes over the time intervals $[k^*-\tau_{\ell,2},k^*-\tau_{\ell,2}+1,\ldots,k^*-\tau_{\ell,2} + \tau_{2,3} -1]$, $\ell\in\{0,1\}$.
In practice, however, the information sent by the exosystem and \textbf{Agent~$\boldsymbol{1}$} reaches \textbf{Agent~$\boldsymbol{2}$} only after $\tau_{0,2}$ and $\tau_{1,2}$ time steps, respectively. Thus, at time $k=k^*$,                      \textbf{Agent~$\boldsymbol{2}$} can only access the exosystem information available at $k^*-\tau_{0,2}$ and the information from \textbf{Agent~$\boldsymbol{1}$} available at $k^*-\tau_{1,2}$.  Consequently, the prediction $\Xi_{2,3}(k)=\xi_2(k+\tau_{2,3})$ is unavailable at time $k=k^*$. This implies that the upper-layer nodes must additionally provide future predicted information over the window $[k, k+1, \ldots, k+\tau_{2,3}-1]$.}

\textcolor{black}{To identify precisely the required additional future information for exact prediction, we first consider the ideal predictor states 
$\Xi_{i,j}(k) =	\xi_i(k + \tau_{i,j})$, $i=1,2$, $i<j$ in order to   determine the structure of the  information transmission flow needed for the design of  implementable predictors over an extended prediction horizon.  From  \eqref{equ:example_pre_1}--\eqref{equ:example_pre_3}, the following requirements can be identified:
\begin{itemize}[leftmargin=*, itemsep=0pt]
	\item First, according to \eqref{equ:example_pre_1}--\eqref{equ:example_pre_3}, the exosystem must  additionally generate the following future information to enable the design of  the predictors for \textbf{Agent~$\boldsymbol{1}$} and \textbf{Agent~$\boldsymbol{2}$}:
	\begin{itemize}[leftmargin=0pt]
		\item[] \textbf{Agent~$\boldsymbol{1:}$} $\Upsilon_1(k + 1)$, $\Upsilon_1(k + 2)$, $\cdots$, $\Upsilon_1(k + \max\{\tau_{1,2}, \tau_{1,3}  \} $ $ - 1)$;
		\item[] \textbf{Agent~$\boldsymbol{2:}$} $\Upsilon_2(k + 1)$, $\Upsilon_2(k + 2)$, $ \cdots$, $\Upsilon_2(k +\tau_{2,3} - 1)$;
	\end{itemize}
	\item Second, according to \eqref{equ:example_pre_3},  \textbf{Agent~$\boldsymbol{1}$} must  additionally generate the following information to allow for the design of   the predictor for \textbf{Agent~$\boldsymbol{2:}$}
	\begin{align*}
		\Xi_{{1,2 }}(k + 1),\ 	\Xi_{1,2}(k + 2),\ \cdots, \Xi_{1,2}(k + \tau_{2,3} - 1).
	\end{align*}
\end{itemize}
Moreover, from the distributed observer dynamics of $\xi_1(k)$ and the ideal predictor $\Xi_{1,2}(k)$, we have
\begin{align*}
	\Xi_{1,2}(k + 1) = \textstyle  \hat{S}_1 \Xi_{1,2}(k)  + \beta a_{1,0} S \Upsilon_1(k -\tau_{0,1} + \tau_{1,2} ).
\end{align*} 
Hence, generating $\Xi_{1,2}(k+1)$ further requires the exosystem to provide the predicted value 
$\Upsilon_1(k-\tau_{0,1}+\tau_{1,2})$. 
In addition, since the ideal predictor $\Xi_{2,3}(k)$ depends on $\Xi_{1,2}(k+s)$, $s=1,2,\ldots,\tau_{2,3}-1$, repeated iteration gives
\begin{align*}
	&\ \Xi_{1,2}(k+s) \notag\\
	=&\ \textstyle \hat{S}_1^{s}\Xi_{1,2}(k) + \beta  \sum_{j=1}^{s} \hat{S}_1^{j-1} a_{1,0} S \Upsilon_1(k - \tau_{0,1} + \tau_{1,2} + s -  j),
\end{align*} 
Therefore, the computation of $\Xi_{1,2}(k+s)$ requires the exosystem predictions over the relative horizons 
$[\tau_{1,2},\tau_{1,2}+1,\ldots,\tau_{1,2}+s-1]$, for $s=1,2,\ldots,\tau_{2,3}-1$.  Combining these requirements, the additional future information that should be provided by exosystem and \textbf{Agent~$\boldsymbol{1}$} is summarized as follows:
\begin{itemize}
	\item Exosystem:
	\begin{align*}
		&\ \Upsilon_1(k + 1),\ 	\Upsilon_1(k + 2),\ \cdots,\\
		&\  \Upsilon_1(k + \max\{ \tau_{1,2} + \tau_{2,3}-2 , \tau_{1,3} -1 \} );\\
		&\ \Upsilon_2(k + 1),\ 	\Upsilon_2(k + 2),\ \cdots, \Upsilon_2(k +\tau_{2,3} - 1).
	\end{align*}
	\item \textbf{Agent~$\boldsymbol{1:}$} $\Xi_{1,2}(k + 1)$, $\Xi_{1,2}(k + 2)$, $\cdots$, $\Xi_{1,2}(k + \tau_{2,3} - 1)$.
\end{itemize}}

\textcolor{black}{To generate the above additional future information in a causal and implementable manner, we introduce the following distributed predictors:
\begin{align*}
	\Upsilon_{i,s}(k) =&\ S\Upsilon_{i,s-1}(k),\quad i=1,2,\\
	\Xi_{1,2,s}(k) =&\ \textstyle \hat{S}_1\Xi_{1,2,s-1}(k) + \beta a_{1,0} S \Upsilon_{1,\tau_{1,2} + s -  1}(k - \tau_{0,1}),
\end{align*} 
where $\Upsilon_{i,s}(k) = \Upsilon_i(k+s)$ and $\Xi_{1,2,s}(k) = \Xi_{1,2}(k+s)$ with $\Upsilon_{i} = \Upsilon_{i,0}(k)$ and $\Xi_{1,2,0}(k) = \Xi_{1,2}(k)$. Accordingly, the predictors in \eqref{equ:example_pre_1}--\eqref{equ:example_pre_3} can be rewritten in the following causal form:
\begin{align*}
	\Xi_{1,2}(k) =&\ \textstyle
	\hat{S}_1^{\tau_{1,2}}\xi_1(k) + \beta \sum_{j=1}^{\tau_{1,2}}  \hat{S}_1^{j-1} a_{1,0} S \Upsilon_{1,\tau_{1,2} - j}(k - \tau_{0,1}), \\
	\Xi_{1,3}(k) 
	=&\ \textstyle
	\hat{S}_1^{\tau_{1,3}}\xi_1(k) + \beta \sum_{j=1}^{\tau_{1,3}}  \hat{S}_1^{j-1} a_{1,0} S \Upsilon_{1,\tau_{1,3} - j}(k - \tau_{0,1}), \\
	\Xi_{2,3}(k) =&\ \textstyle
	\hat{S}_2^{\tau_{2,3}}\xi_2(k) + \beta \sum_{j=1}^{\tau_{2,3}}  \hat{S}_2^{j-1} a_{2,0} S \Upsilon_{2,\tau_{2,3} - j }(k - \tau_{0,2}) \notag\\
	& \textstyle + \beta \sum_{j=1}^{\tau_{2,3}} \hat{S}_2^{j-1} a_{2,1} S \Xi_{1,2,\tau_{2,3} - j}(k - \tau_{1,2} ).
\end{align*}
It remains to develop a direct method for determining the required range of the prediction index $s$. 
For this purpose, \eqref{equ:prediction_horizon_heter} in Definition \ref{def1}, provides the extended prediction horizon, which specifies the required range of $s$. In this example, the  exosystem needs to provide information $\Upsilon_{1,s}$, $s\in\mathbb{P}_{0,1} = \{0,1,\ldots, w_{1,3}\}$ to \textbf{Agent~$\boldsymbol{1}$} and information $\Upsilon_{2,s}$, $s\in\mathbb{P}_{0,2} =\{0,1,\ldots,w_{2,3}\}$ to \textbf{Agent~$\boldsymbol{2,}$} where $w_{1,3}$, $w_{2,3}$ are defined in \eqref{equ: longest_path}. In addition, \textbf{Agent~$\boldsymbol{1}$} also needs to provide information $\Xi_{1,2}(k+s)$, $s\in\mathbb{P}_{1,2} = \{0,1,\ldots,w_{2,3}\}$ to \textbf{Agent~$\boldsymbol{2.}$}}
\begin{figure}[!h]
	\centering
	\begin{tikzpicture}[scale=0.4]
		\draw[fill=Maroon!50,very thick] (-2 ,0) circle (0.5) node {\scriptsize $0$};
		\draw[fill=NavyBlue!50,very thick] (0,0) circle (0.5) node {\scriptsize $1$};
		
		\draw[fill=NavyBlue!50,very thick] (2,0) circle (0.5) node {\scriptsize $2$};
		
		\draw[->, dashed, Blue, very thick, >=latex]  (-1.5,0) ->  (-0.5,0) ;
		
		\draw[->, dashed, Blue, very thick, >=latex]  (0,0.5) to [out=30,in=150] (2,0.5);
		\draw[->, dashed, Blue, very thick, >=latex]  (2,-0.5) to [out=-150,in=-30] (0,-0.5);
	\end{tikzpicture}
	\caption{\textcolor{black}{Communication topology with a directed cycle.}}
	\label{fig_structure2}
\end{figure}

\begin{remark}[On directed cycles graphs]
	\textcolor{black}{The feasibility of the proposed prediction procedure relies on Assumption~\ref{assum1}.  If the communication topology contains directed cycles, exact prediction may result in an infinite forward dependency.  Consider the two-agent topology in Fig.~\ref{fig_structure2} with uniform communication delays $\tau_{i,j}=2$. Following \cite{huang2016cooperative}, the distributed observer is given by:
	\begin{align}	
		\xi_1(k+1) =&\ \textstyle S\xi_1(k) - \beta S a_{1,0} (\xi_1(k) - \Upsilon_{1}(k-2)) \notag\\
		&\textstyle - \beta S a_{1,2} (\xi_1(k) - \Xi_{2,1}(k-2)), \label{equ:x1}\\
		\xi_2(k+1) =&\ \textstyle S\xi_2(k) -  \beta Sa_{2,1} (\xi_2(k) - \Xi_{1,2}(k-2)), \label{equ:x2}
	\end{align}
	 Let $\Xi_{1,2}(k)=\xi_1(k+2)$ and $\Xi_{2,1}(k)=\xi_2(k+2)$ be the ideal predictors. 
From \eqref{equ:x1} and \eqref{equ:x2}, one obtains
		\begin{align*}	
		\Xi_{1,2}(k) =&\ \textstyle \hat{S}_1^2\xi_1(k) + \beta  \hat{S}_1 a_{1,0} S \Upsilon_{1}(k-2) + \beta \hat{S}_1 a_{1,2}S\Xi_{2,1}(k-2) \\
		& \textstyle + \beta  a_{1,0} S \Upsilon_{1}(k + 1 - 2) + \beta  a_{1,2} S \Xi_{2,1}(k + 1 - 2),\\
		\Xi_{2,1}(k) = &\ \textstyle \hat{S}_2^2\xi_2(k) + \beta\hat{S}_2a_{2,1}S\Xi_{1,2}(k-2)  \\
        & + \beta a_{2,1}S\Xi_{1,2}(k + 1 - 2),
	\end{align*}
	where $\hat{S}_1 = S - \beta a_{1,0}S - \beta a_{1,2} S$ and $\hat{S}_2 = S - \beta a_{2,1} S$. Thus, $\Xi_{1,2}(k)$ depends on $\Xi_{2,1}(k + 1 - 2)$, while $\Xi_{2,1}(k)$ depends on $\Xi_{1,2}(k + 1 - 2)$.  Due to the communication delays, however, these two terms are not directly available to \textbf{Agent~$\boldsymbol{1}$}  and \textbf{ Agent~$\boldsymbol{2,}$} respectively.  To generate the missing terms recursively, one would need
	\begin{align*}
		\Xi_{1,2}(k+1) = &\ \textstyle \hat{S}_1\Xi_{1,2}(k ) + \beta  a_{1,0} S \Upsilon_{1}(k + 2- 2) \\
		& \textstyle  + \beta  a_{1,2} S \Xi_{2,1}(k + 2 - 2),\\
		\Xi_{2,1}(k + 1) =&\ \textstyle \hat{S}_2 \Xi_{2,1}(k) + \beta a_{2,1}S\Xi_{1,2}(k + 2 - 2).
	\end{align*} 
	Now, to obtain $\Xi_{1,2}(k)$, \textbf{Agent~$\boldsymbol{1}$} requires \textbf{Agent~$\boldsymbol{2}$} to send $\Xi_{2,1}(k+1)$ in advance. However, $\Xi_{2,1}(k+1)$ itself depends on $\Xi_{1,2}(k+2 -2)$, which in turn requires \textbf{Agent~$\boldsymbol{1}$} to send $\Xi_{1,2}(k+2)$ in advance. By repeating this argument, an infinite forward dependence arises. If any of the required information is unavailable, exact delay compensation fail to be achieved. }
\end{remark}
\begin{remark}\textcolor{black}{Although Assumption~\ref{assum1} is restrictive to some extent, it can be accommodated in practical implementation. In general, the communication topology is assumed to contain a directed spanning tree, including strongly connected topologies. One may select an acyclic subgraph that preserves the required information flow and ignore the remaining edges. From this perspective, the proposed algorithm is applicable to all the aforementioned communication topologies. This strategy is feasible for most connected directed communication topologies, because the existence of a directed acyclic subgraph is a mild requirement for a connected topology.}
\end{remark}

\subsection{Generalization of the prediction-based distributed observer design}
\textcolor{black}{ Following the arguments presented in Section \ref{infeas},  we now generalize the prediction-based distributed observer design.}

Firstly, to satisfy the design requirements of each agent, the exosystem \eqref{equ:exosystem} should additionally construct the  predictor:
\begin{align}
	\Upsilon_i(k) =&\ \textstyle \upsilon(k+\tau_{0,i}) = S^{\tau_{0,i}} \upsilon(k), \label{equ:predictor_upsilon_heter}
\end{align}
and distributed predictor:
\begin{align}
	\Upsilon_{i,s}(k) =&\ \textstyle  \Upsilon_{i}(k+s) = S^s\Upsilon_{i}(k), \label{equ:distributed_predictor_upsilon_s_heter}
\end{align}
\textcolor{black}{where the predictor $\Upsilon_i(k)\in\mathbb{R}^q$ denotes the $\tau_{0,i}$-step-ahead prediction of $\upsilon(k)$ and the distributed predictor $\Upsilon_{i,s}(k)\in\mathbb{R}^q$, $s\in\mathbb{P}_{0,i}$, denotes the $s$-step-ahead prediction of $\Upsilon_i(k)$. Both the predictor $\Upsilon_i(k)$ and distributed predictor $\Upsilon_{i,s}(k)$ are constructed by the exosystem and  transmitted to \textbf{Agent} $\boldsymbol{i}$. In particular, $\Upsilon_{i,0}(k) = \Upsilon_{i}(k)$.}

Furthermore, for \textbf{Agent} $\boldsymbol{i}$, $i=1,2,\ldots,N$,  the following distributed observers are constructed:
\begin{align}	
	\xi_i(k+1) =&\ \textstyle S\xi_i(k) -  \beta S a_{i,0} (\xi_i(k) - \Upsilon_i(k-\tau_{0,i} )) \notag\\
	&  \textstyle - \beta S \sum_{j=1}^{i-1}a_{i,j} (\xi_i(k) - \Xi_{j,i}(k-\tau_{j,i})) \notag\\
	=&\ \textstyle \hat{S}_i\xi_i(k) + \beta  a_{i,0} S \Upsilon_i(k-\tau_{0,i} ) \notag \\
	& \textstyle + \beta  \sum_{j=1}^{i-1}a_{i,j} S \Xi_{j,i}(k-\tau_{j,i}). \label{equ:distributed_ob_heter}
\end{align}
together with the predictor for \textbf{Agent} $\boldsymbol{i}$, $\boldsymbol{i}=1,2,\ldots,N-1$
\begin{align}
	\Xi_{i,r}(k) =&\ \textstyle \hat{S}_i^{\tau_{i,r}}\xi_i(k) +  \beta \sum_{j = 1}^{\tau_{i,r}}\hat{S}_i^{\tau_{i,r}-j}  a_{i,0} S\Upsilon_{i,j-1}(k-\tau_{0,i}) \notag\\
	& \textstyle +  \beta \sum_{j=1}^{i-1}  \sum_{\ell=1}^{\tau_{i,r}}\hat{S}_i^{\tau_{i,r}-\ell} a_{i,j} S \Xi_{j,i,\ell-1}( k-\tau_{j,i}), \label{equ:predictor_xi_heter}
\end{align}
and the distribtued predictor for \textbf{Agent} $\boldsymbol{i}$, $\boldsymbol{i}=1,2,\ldots,N-1$
\begin{align}
	\Xi_{i,r,s}(k) =&\ \textstyle \hat{S}_{i}\Xi_{i,r,s-1}(k) +  \beta a_{i,0} S \Upsilon_{i,\tau_{i,r} + s-1}(k-\tau_{0,i}) \notag \\
	& \textstyle + \beta \sum_{j=1}^{i-1} a_{i,j} S \Xi_{j,i,\tau_{i,r} + s-1}(k-\tau_{j,i}), \label{equ:distributed_predictor_xi_s_heter}
\end{align}
where  $\xi_i(k)\in\mathbb{R}^q$ represents the distributed observer state,  \textcolor{black}{the predictor  $\Xi_{j,i}(k)\in\mathbb{R}^{q}$ denotes the $\tau_{j,i}$-step-ahead prediction of $\xi_j(k)$ and $\Xi_{j,i,s}\in\mathbb{R}^q$, $s\in\mathbb{P}_{j,i}$, denotes the $s$-step-ahead prediction of $\Xi_{j,i}(k)$. Both the predictor $\Xi_{j,i}(k)$ and the distributed predictor $\Xi_{j,i,s}(k)$ are constructed by the \textbf{Agent} $\boldsymbol{j}$, $\boldsymbol{j}=1,2,\ldots,i-1$ and transmitted to \textbf{Agent} $\boldsymbol{i}$.} We recall that $\beta$ is the coupling gain and $\hat{S}_i = S - \beta\sum_{j=0}^{i-1}a_{i,j}S$. In particular, $\Xi_{i,r,0}(k) = \Xi_{i,r}(k)$.

Additionally, due to the existence of communication delay, when $k < \tau_{j,i}$, the information of \textbf{Agent} $\boldsymbol{j}$ has not been delivered to \textbf{Agent} $\boldsymbol{i}$ yet. 
Therefore, we assume $\Upsilon_{i}(k)=0$, $\Upsilon_{i,s}(k)=0$, $\Xi_{i,r}(k) = 0$, $\Xi_{i,r,s}(k) = 0$, for $k<0$. Consequently, for $k < \min\{\tau_{\ell,i}|\ell=0,1,\ldots,i-1\}$, the dynamics of distributed observers, predictors and distributed predictors for \textbf{Agent} $\boldsymbol{i}$ are transformed as follows: 
\begin{gather}	
	\begin{array}{c}
	\xi_i(k+1) = \textstyle  \hat{S}_i\xi_i(k), \\
		\Xi_{i,r}(k) = \hat{S}_i^{\tau_{i,r}}\xi_i(k),\	\Xi_{i,r,s}(k) =  \hat{S}_{i}\Xi_{i,r,s-1}(k).
	\end{array}  \label{equ:without_information}
\end{gather}

\begin{table*}[!t]
\centering
\footnotesize
\setlength{\tabcolsep}{4.5pt}
\renewcommand{\arraystretch}{1.25}
\caption{\textcolor{black}{Communication load of the proposed prediction-based algorithm under full-prediction-window transmission.}}
\label{tab:communication_load}
\begin{tabular}{@{}
>{\raggedright\arraybackslash}p{0.09\textwidth}
>{\raggedright\arraybackslash}p{0.34\textwidth}
>{\raggedright\arraybackslash}p{0.255\textwidth}
>{\raggedright\arraybackslash}p{0.26\textwidth}
@{}}
\toprule
Node & Received information & Sent information & Number of $q$-dimensional arrays \\
\midrule
Exosystem
&
None
&
For each $(0,i)\in\mathcal E$: 
$\{\Upsilon_{i,s}(k)\}_{s\in P_{0,i}}$.
&
Received: $0$.

Sent: $(w_{i,N}+1)$ per edge $(0,i)$.
\\
\midrule
Agent $i$, $1\leq i\leq N$
&
If $(0,i)\in\mathcal E$: 
$\{\Upsilon_{i,s}(k-\tau_{0,i})\}_{s\in P_{0,i}}$.

For each $(j,i)\in\mathcal E$, $j < i$: 
$\{\Xi_{j,i,s}(k-\tau_{j,i})\}_{s\in P_{j,i}}$.
&
For each $(i,r)\in\mathcal E$: 
$\{\Xi_{i,r,s}(k)\}_{s\in P_{i,r}}$.
&
Received: $d_i^-(w_{i,N}+1)$.

Sent: $\sum_{r\in\mathcal N_i^+}(w_{r,N}+1)$.
\\
\bottomrule
\end{tabular}

\vspace{1mm}
\begin{tablenotes}
\footnotesize
\item[] Note: Here, $\mathcal N_i^-=\{j\in\{0,1,\ldots,i-1\}:a_{i,j}\neq0\}$, $\mathcal N_i^+=\{r\in\{i+1,\ldots,N\}:a_{r,i}\neq0\}$, and $d_i^-$ denotes the cardinality of $\mathcal N_i^-$. 
\end{tablenotes}
\end{table*}

\begin{figure}[!t]
\centering
\resizebox{0.47\textwidth}{!}{
\begin{tikzpicture}[>=latex, line width=0.8pt, every node/.style={font=\small}]

\def\yexo{2}
\def\yj{1}
\def\yi{0}

\def\xstart{0}
\def\xend{5.2}

\draw[->] (\xstart,\yexo) -- (\xend,\yexo);
\draw[->] (\xstart,\yj)   -- (\xend,\yj);
\draw[->] (\xstart,\yi)   -- (\xend,\yi);

\node[left] at (\xstart,\yexo) {\tiny Exosystem};
\node[left] at (\xstart,\yj)   {\tiny Agent $j$};
\node[left] at (\xstart,\yi)   {\tiny Agent $i$};

\foreach \y in {\yexo,\yj,\yi}{
    \draw (0,\y+0.045) -- (0,\y-0.045);
}

\node[below] at (0,\yi) {\tiny $k$};

\def\xtauo{1}   
\def\xtauj{3.2}   

\foreach \y in {\yexo,\yj,\yi}{
    \draw (\xtauj,\y+0.045) -- (\xtauj,\y-0.045);
    \draw (\xtauo,\y+0.045) -- (\xtauo,\y-0.045);
}

\node[below] (taujlabel) at (\xtauj,\yi) {\tiny $k+\tau_{j,i}$};
\node[below] (tauolabel) at (\xtauo,\yi) {\tiny $k+\tau_{0,i}$};

\def\xso{2.2}   
\def\xsj{4.4}   

\foreach \y in {\yexo,\yj,\yi}{
    \draw (\xsj,\y+0.045) -- (\xsj,\y-0.045);
    \draw (\xso,\y+0.045) -- (\xso,\y-0.045);
}

\node[below] at (\xsj,\yi) {\tiny $k+\tau_{j,i}+s$};
\node[below] at (\xso,\yi) {\tiny $k+\tau_{0,i}+s$};

\node[below] (upsilonlabel) at (0,\yexo) {\tiny $\upsilon(k)$};
\node[below] (xijlabel) at (0,\yj)        {\tiny $\xi_j(k)$};

\node[below] (Upsilonlabel) at (\xtauo,\yexo) {\tiny $\Upsilon_i(k)$};
\node[below] (Xijlabel) at (\xtauj,\yj)   {\tiny $\Xi_{j,i}(k)$};

\node[below] (Upsilonslabel) at (\xso,\yexo) {\tiny $\Upsilon_{i,s}(k)$};
\node[below] (Xijslabel) at (\xsj,\yj)   {\tiny $\Xi_{j,i,s}(k)$};

\draw[->,blue, thick] (xijlabel.east) to (Xijlabel.west);
\draw[->,blue, thick] (Xijlabel.east) to (Xijslabel.west);
\draw[->,blue, thick] (upsilonlabel.east) to (Upsilonlabel.west);
\draw[->,blue, thick] (Upsilonlabel.east) to (Upsilonslabel.west);

\node[
    draw,
    dashed, red,
    rounded corners=3pt,
    inner xsep=4pt,
    inner ysep=3pt,
    fit=(Xijlabel)(Xijslabel)
] (Xigroup) {};

\draw[->, dashed, red, thick] 
(Xigroup.south) .. controls +(0,-0.35) and +(0,0.35) .. node[pos=0.55,  right=8pt, red] {\tiny $\tau_{j,i}$} (taujlabel.north);

\node[
    draw,
    dashed, red,
    rounded corners=3pt,
    inner xsep=5pt,
    inner ysep=4pt,
    fit=(Upsilonlabel)(Upsilonslabel)
] (Ugroup) {};

\draw[->, dashed, red, thick] 
(Ugroup.south) .. controls +(0,-0.35) and +(0,0.25) .. node[pos=0.8, right, red] {\tiny $\tau_{0,i}$} (tauolabel.north);

\end{tikzpicture}
}
\caption{\textcolor{black}{Timeline of the prediction transmission mechanism.}}
\label{fig_timeline_basic}
\end{figure}
 \textcolor{black}{In Fig. \ref{fig_timeline_basic}, we present the timeline of the information flow. The three horizontal lines represent the timelines of the exosystem, \textbf{Agent} $\boldsymbol{j}$, and \textbf{Agent} $\boldsymbol{i}$, respectively. The blue solid arrows indicate the dependency relationships among the variables. The labels are placed at the time instants corresponding to the states they predict. The red dashed boxes indicate the groups of signals that are transmitted together, while the red dashed arrows indicate their arrival instants at \textbf{Agent} $\boldsymbol{i}$. In particular, $\Upsilon_i(k)$ and $\Upsilon_{i,s}(k)$ are generated by the exosystem at time $k$ and sent together to \textbf{Agent} $\boldsymbol{i}$, arriving at time $k+\tau_{0,i}$. Similarly, $\Xi_{j,i}(k)$ and $\Xi_{j,i,s}(k)$ are generated by \textbf{Agent} $\boldsymbol{j}$ at time $k$ and sent together to \textbf{Agent} $\boldsymbol{i}$, arriving at time $k+\tau_{j,i}$. It should be noted that the ordering $\tau_{0,i}<\tau_{j,i}$ in Fig. \ref{fig_timeline_basic} is only adopted for a clearer graphical illustration as the  proposed method does not require such an ordering.}

 \textcolor{black}{\textbf{Communication overhead.} Furthermore, we provide an explicit analysis of the communication overhead of the proposed algorithm. The information received and sent by each node under the full-prediction-window transmission mechanism is summarized in Table~\ref{tab:communication_load}. Specifically, according to the definition of the extended prediction horizon in \eqref{equ:prediction_horizon_heter}, $P_{j,i}=\{0,1,\ldots,w_{i,N}\}$.  Therefore, each edge $(j,i)$ transmits exactly $w_{i,N}+1$ arrays, each of dimension $q$. More precisely, the exosystem sends $\{\Upsilon_{i,s}(k)\}_{s\in P_{0,i}}$. Upper-layer \textbf{Agent} $\boldsymbol{j}$, $\boldsymbol{j}\in\{1,2,N-1\}$, sends $\{\Xi_{j,i,s}(k)\}_{s\in P_{j,i}}$. Then the number of $q$-dimensional arrays received by \textbf{Agent} $\boldsymbol{i}$ is $d_i^-(w_{i,N}+1)$. Similarly, the number of $q$-dimensional arrays sent by \textbf{Agent} $\boldsymbol{i}$ is $\sum_{r\in\mathcal N_i^+}(w_{r,N}+1)$. Consequently, the maximum network communication load is $\mathcal C_{\rm max} = q\sum_{i=1}^N d_i^-(w_{i,N}+1)$.}

\textcolor{black}{Then, we discuss how the communication load grows with
the size of the graph. Let
$\tau_{\max}=\max_{(j,i)\in\mathcal E}\tau_{j,i}$,
and let $\mathbb{L}_{\max}$ denote the maximum number of edges in a directed path
of the DAG. Since the modified longest-path weight satisfies
$w_{i,N}\leq (\tau_{\max}-1)\mathbb{L}_{\max}$,
the maximum communication load satisfies
$\mathcal C_{\rm max} \leq q|\mathcal E|(1+(\tau_{\max}-1)\mathbb{L}_{\max}) = O(q|\mathcal{E}|\tau_{\max}\mathbb{L}_{\max})$.
Thus, the maximum overhead grows linearly with the dimension $q$ of the exosystem, 	linearly with the number of communication edges $|\mathcal E|$, and linearly with the maximum communication delay $\tau_{\max}$, and linearly with the graph depth $\mathbb{L}_{\max}$. For a sparse chain-like graph, where $|\mathcal E|=O(N)$ and $\mathbb{L}_{\max}=O(N)$, the worst-case load is $O(q\tau_{\max}N^2)$. For a dense DAG, where $|\mathcal E|=O(N^2)$ and $\mathbb{L}_{\max}=O(N)$, the worst-case load is 	$O(q\tau_{\max}N^3)$. For shallow DAGs with bounded depth, the communication load 	remains $O(q\tau_{\max}|\mathcal E|)$.}

\section{Exact compensation of communication delays and distributed observers consensus}\label{adsec}
\subsection{Exact compensation property of predictors}
\textcolor{black}{In the following, we present a technical lemma to show that the constructed predictors \eqref{equ:predictor_upsilon_heter}--\eqref{equ:predictor_xi_heter} can \emph{exactly compensate} for the communication delays \emph{after a period of inexact prediction.} This property will be used to prove the consensus of the distributed observers.} To this end, we first define the prediction errors as
\begin{align}
	\tilde{\Upsilon}_{i}(k) =&\ \Upsilon_{i}(k)-\upsilon(k+\tau_{0,i}), \label{equ:prediction_error_u}\\
	 \tilde{\Xi}_{i,j}(k) =&\ \Xi_{i,j}(k)-\xi_i(k+\tau_{i,j}), \label{equ:prediction_error_xi}
\end{align}
and the distributed prediction errors as  
\begin{align}
	\tilde{\Upsilon}_{i,s}(k) =&\ \Upsilon_{i,s}(k)-\upsilon(k  + s +\tau_{0,i}), \label{equ:distribtued_prediction_error_u}  \\
	 \tilde{\Xi}_{i,j,s}(k) =&\ \Xi_{i,j,s}(k)-\xi_i(k  + s +\tau_{i,j}). \label{equ:distributed_prediction_error_xi}
\end{align}
\begin{lemma} \label{lemma1}
	Let Assumption \ref{assum1} hold. Let $T_0=0$ and $T_i=\max_{\pi\in\Pi_{0,i}}\{w(\pi)\}$, $i=1,2,\ldots,N$. The predictors \eqref{equ:predictor_upsilon_heter}, \eqref{equ:predictor_xi_heter} and distributed  predictors \eqref{equ:distributed_predictor_upsilon_s_heter}, \eqref{equ:distributed_predictor_xi_s_heter} satisfy:\\
	\textbf{\rm R}$\boldsymbol{1:}$ $\Upsilon_i(k+s) = \upsilon(k + s + \tau_{0,i}) = \Upsilon_{i,s}(k)$, for $ s \in \mathbb{P}_{0,i}$, $k\ge T_0$;\\
	\textbf{\rm R}$\boldsymbol{2:}$  $\Xi_{i,r}(k + s) = \xi_i(k + s + \tau_{i,r}) = \Xi_{i,r,s}(k)$, for $i=1,2,\ldots,N - 1$, $r=i+1,\ldots,N$ with $a_{r,i}>0$, $s \in \mathbb{P}_{i,r}$, $k\ge T_i$. 
\end{lemma}
\begin{proof}
   \textcolor{black}{ First, we prove result \textbf{\rm R}$\boldsymbol{1.}$   Since the exosystem is autonomous, one has
    \begin{align}
	   \upsilon(k+s+\tau_{0,i}) = S^{s+\tau_{0,i}}\upsilon(k), \quad k\ge T_0.
    \end{align}
    Then, by \eqref{equ:predictor_upsilon_heter} and \eqref{equ:distributed_predictor_upsilon_s_heter},
    \begin{align}
    	\Upsilon_{i,s}(k) = S^s\Upsilon_i(k) = S^{s+\tau_{0,i}}\upsilon(k) = \upsilon(k+s+\tau_{0,i}).
    \end{align}
    Moreover,
    \begin{align}
    	\Upsilon_i(k+s) = S^{\tau_{0,i}}\upsilon(k+s) = \upsilon(k+s+\tau_{0,i}).
    \end{align}
    Hence,
    \begin{align}
    	\Upsilon_i(k+s) = \upsilon(k+s+\tau_{0,i}) = \Upsilon_{i,s}(k),
    \end{align}
    for all $s\in\mathbb P_{0,i}$ and $k\ge T_0$. This proves \textbf{\rm R}$\boldsymbol{1.}$}
    
    \textcolor{black}{Next, we prove result \textbf{\rm R}$\boldsymbol{2}$ by induction according to the topological ordering of the DAG. For any positive integer $h$, iterating
    \eqref{equ:distributed_ob_heter} gives
    \begin{align}
    	\xi_i(k+h)	=&\ \textstyle \hat S_i^h\xi_i(k) +\beta\sum_{\ell=1}^{h}\hat S_i^{h-\ell}a_{i,0}S
    	\Upsilon_i(k+\ell-1-\tau_{0,i}) \notag\\
    	& \textstyle +\beta\sum_{j=1}^{i-1}\sum_{\ell=1}^{h}\hat S_i^{h-\ell}a_{i,j}S\Xi_{j,i}(k+\ell-1-\tau_{j,i}).
    \end{align}
    For any $s\in\mathbb P_{i,r}$, starting from $\Xi_{i,r,0}(k)=\Xi_{i,r}(k)$ and iterating \eqref{equ:distributed_predictor_xi_s_heter} from $1$ to $s$, we obtain
    \begin{align}
    	\Xi_{i,r,s}(k) =&\ \textstyle \hat S_i^{\tau_{i,r}+s}\xi_i(k) + \beta\sum_{\ell=1}^{\tau_{i,r}+s} \hat S_i^{\tau_{i,r}+s-\ell}a_{i,0}S\Upsilon_{i,\ell-1}(k-\tau_{0,i}) \notag\\
    	& \textstyle +\beta\sum_{j=1}^{i-1}\sum_{\ell=1}^{\tau_{i,r}+s}\hat S_i^{\tau_{i,r}+s-\ell}a_{i,j}S\Xi_{j,i,\ell-1}(k-\tau_{j,i}).
    \end{align}
    In particular, for $s=0$, this formula reduces to \eqref{equ:predictor_xi_heter}.
    Taking $h=\tau_{i,r}+s$ in the above expression of $\xi_i(k+h)$ and subtracting it from the above expression of $\Xi_{i,r,s}(k)$ yield
    \begin{align}
    	&\ \tilde{\Xi}_{i,r,s}(k) \notag\\
    	=&\ \textstyle \beta\sum_{\ell=1}^{\tau_{i,r}+s}\hat S_i^{\tau_{i,r}+s-\ell}a_{i,0}S[\tilde{\Upsilon}_{i,\ell-1}(k-\tau_{0,i}) - \tilde{\Upsilon}_{i}(k+\ell-1-\tau_{0,i})] \notag\\
    	& \textstyle +\beta\sum_{j=1}^{i-1}\sum_{\ell=1}^{\tau_{i,r}+s}\hat S_i^{\tau_{i,r}+s-\ell}a_{i,j}S \notag\\
    	& \textstyle \times[\tilde{\Xi}_{j,i,\ell-1}(k-\tau_{j,i}) - \tilde{\Xi}_{j,i}(k+\ell-1-\tau_{j,i})]. \label{equ:Xi_error_relation_lemma1}
    \end{align}}
    \textcolor{black}{\indent We now show that the right-hand side of the above equation vanishes for $k\ge T_i$. Firstly, for any active edge $(j,i)$ with $j<i$, every path from node $0$ to node $j$ followed by the edge $(j,i)$ is a path from node $0$ to node $i$. Thus, it follows that $T_i\ge T_j+\tau_{j,i}$. Furthermore, for $k\ge T_i$, we have $k-\tau_{j,i}\ge T_j$. In addition, if $a_{i,0}>0$, then $T_i\ge \tau_{0,i}$, and hence one has $k-\tau_{0,i}\ge T_0$.}
    
    \textcolor{black}{ Subsequently, we prove the claim by induction according to the topological ordering.
    For the base case, consider \textbf{Agent} $\boldsymbol{1.}$ Since there is no agent indexed before \textbf{Agent} $\boldsymbol{1,}$ the summation over $j=1,\ldots,i-1$ is empty. Moreover, for $k\ge T_1$, one has $k-\tau_{0,1}\ge T_0$. Hence, by \textbf{R}$\boldsymbol{1,}$ all exosystem-related error terms in Eq. \eqref{equ:Xi_error_relation_lemma1} vanish. Therefore, 
    \begin{align}
    	\tilde{\Xi}_{1,r,s}(k)=0, \qquad k\ge T_1,\quad s\in\mathbb P_{1,r}.
    \end{align}
    Now assume that the claim holds for all agents $\{1,2,\ldots,i-1\}$. We prove that it holds for \textbf{Agent} $\boldsymbol{i}.$ For $k\ge T_i$, it follows from $T_i\ge T_j+\tau_{j,i}$ that
    \begin{align}
    	k-\tau_{j,i}\ge T_j
    \end{align}
    for every active upper-layer neighbor $j$ of \textbf{Agent} $\boldsymbol{i}$ 
    Since $\ell\ge1$, one also has
    \begin{align}
    	k+\ell-1-\tau_{j,i}\ge T_j.
    \end{align}
    To apply the induction hypothesis to $\tilde{\Xi}_{j,i,\ell-1}(k-\tau_{j,i})$, we need to verify that the prediction index $\ell-1$ lies within the extended prescribed horizon $\mathbb P_{j,i}$. Indeed, $\ell-1$ is admissible because, for $s\in\mathbb P_{i,r}$,
    \begin{align}
    	\ell-1 \le \tau_{i,r}+s-1 \le \tau_{i,r}+w_{r,N}-1 \le w_{i,N}.
    \end{align}
    Thus, by the induction hypothesis and by the fact that
    $\tilde{\Xi}_{j,i,0}(t)=\tilde{\Xi}_{j,i}(t)$, we have
    \begin{align}
    	\tilde{\Xi}_{j,i,\ell-1}(k-\tau_{j,i})=0,\quad \tilde{\Xi}_{j,i}(k+\ell-1-\tau_{j,i})=0.
    \end{align}
    Similarly, by result \textbf{\rm R}$\boldsymbol{1,}$ the exosystem error terms satisfy
    \begin{align}
    	\tilde{\Upsilon}_{i,\ell-1}(k-\tau_{0,i})=0, \quad \tilde{\Upsilon}_{i}(k+\ell-1-\tau_{0,i})=0,
    \end{align}
    whenever $a_{i,0}>0$; if $a_{i,0}=0$, the corresponding terms vanish trivially. Substituting the above identities into \eqref{equ:Xi_error_relation_lemma1} gives
    \begin{align}
    	\tilde{\Xi}_{i,r,s}(k)=0, \quad k\ge T_i,\quad s\in\mathbb P_{i,r}.
    \end{align}
    Thus, by topological induction, the above identity holds for all $i=1,2,\ldots,N-1$.}
    
    \textcolor{black}{ By the definition of $\tilde{\Xi}_{i,r,s}(k)$, this further implies
    \begin{align}
    	\Xi_{i,r,s}(k) = \xi_i(k+s+\tau_{i,r}), \quad k\ge T_i,\quad s\in\mathbb P_{i,r}.
    \end{align}
    Taking $s=0$ gives
    \begin{align}
    	\Xi_{i,r}(k) = \xi_i(k+\tau_{i,r}), \quad k\ge T_i.
    \end{align}
    Since $k\ge T_i$ implies $k+s\ge T_i$, applying the last identity at time
    $k+s$ yields
    \begin{align}
    	\Xi_{i,r}(k+s) = \xi_i(k+s+\tau_{i,r}).
    \end{align}
    Consequently,
    \begin{align}
    	\Xi_{i,r}(k+s) = \xi_i(k+s+\tau_{i,r}) = \Xi_{i,r,s}(k),
    \end{align}
    for all $i=1,2,\ldots,N-1$, $r=i+1,\ldots,N$ with $a_{r,i}>0$, $s\in\mathbb P_{i,r}$, and $k\ge T_i$. This proves result \textbf{R}$\boldsymbol{2}$ and completes the proof.}
\end{proof}
\begin{remark}
    It is worth noting why the condition $k\ge T_i$ is used. The state $\xi_i(k+h)$ is obtained by iterating the distributed observer from time $k$ to time $k+h-1$. At the $\ell$-th step, i.e., at time $k+\ell-1$, the distributed observer uses the delayed information
    \begin{align}
    	\Upsilon_i(k+\ell-1-\tau_{0,i}),\quad \Xi_{j,i}(k+\ell-1-\tau_{j,i}).
    \end{align}
    These signals are indeed available to \textbf{Agent} $\boldsymbol{i}$  at the time $k+\ell-1$. However, the predictor $\Xi_{i,r}(k)$ and the distributed predictor $\Xi_{i,r,s}(k)$ are computed at the current time $k$ to predict the corresponding future state. Hence, \textbf{Agent} $\boldsymbol{i}$ cannot directly use the future-arriving signals $\Upsilon_i(k+\ell-1-\tau_{0,i})$ and $\Xi_{j,i}(k+\ell-1-\tau_{j,i})$ when constructing the prediction at time
    $k$. Instead, it replaces them by the information already available at time $k$, namely
    \begin{align}
    	\Upsilon_{i,\ell-1}(k-\tau_{0,i}),\quad	\Xi_{j,i,\ell-1}(k-\tau_{j,i}).
    \end{align}
    Therefore, the prediction is exact only if the following replacement relations hold:
    \begin{align}
    	\Upsilon_{i,\ell-1}(k-\tau_{0,i}) =	\Upsilon_i(k+\ell-1-\tau_{0,i}), \label{equ:replacement_upsilon}
    \end{align}
    and
    \begin{align}
    	\Xi_{j,i,\ell-1}(k-\tau_{j,i}) = \Xi_{j,i}(k+\ell-1-\tau_{j,i}). \label{equ:replacement_xi}
    \end{align}
    
    For the exosystem-related term, \eqref{equ:replacement_upsilon} is guaranteed by \textbf{\rm R}$\boldsymbol{1}$ in Lemma \ref{lemma1} whenever
    $k-\tau_{0,i}\ge T_0=0$.
    For the upper-layer term, \eqref{equ:replacement_xi} is guaranteed by the corresponding result for \textbf{Agent} $ \boldsymbol{j}$ whenever $k-\tau_{j,i}\ge T_j$. Thus, in order to guarantee that all replacement relations used by \textbf{Agent} $ \boldsymbol{i}$ are exact, it is sufficient to require
    $k\ge T_j+\tau_{j,i}$ for every active upper-layer neighbor $j$ of \textbf{Agent} $ \boldsymbol{i},$ and also $k\ge\tau_{0,i}$ whenever $a_{i,0}>0$. By the definition of $T_i$, these conditions are simultaneously satisfied when $k\ge T_i$.
    
    If $k<T_i$, then there exists a path $\pi^\star$ from node $0$ to node $i$ such that $w(\pi^\star)=T_i$. Let $(j^\star,i)$ be the last edge of this path. Then $T_i=T_{j^\star}+\tau_{j^\star,i}$, and hence $k-\tau_{j^\star,i}<T_{j^\star}$. Therefore, the replacement relation associated with this critical predecessor, $\Xi_{j^\star,i,\ell-1}(k-\tau_{j^\star,i}) = \Xi_{j^\star,i}(k+\ell-1-\tau_{j^\star,i})$ or, when $j^\star=0$, $\Upsilon_{i,\ell-1}(k-\tau_{0,i}) = \Upsilon_i(k+\ell-1-\tau_{0,i})$
    is not guaranteed to hold. Consequently, the corresponding bracketed term in \eqref{equ:Xi_error_relation_lemma1} may be nonzero. Thus, exactness before $T_i$ cannot be guaranteed in general.
\end{remark}
\begin{remark}[Load after exact compensation] 
    \textcolor{black}{ Lemma~\ref{lemma1} shows that the proposed predictors achieve exact delay compensation after a finite time step. During this transient interval, the full prediction windows in Table~\ref{tab:communication_load} are transmitted. After exact compensation has been achieved, the prediction windows can be updated in a shift-register manner. In this implementation, previously transmitted prediction values are reused by shifting the prediction window, and only the newly generated terminal prediction needs to be transmitted over each communication edge. Therefore, the maximum communication load is reduced to $q|\mathcal E|$. Thus, the communication load remains at its maximum value during the finite full-window transmission phase and then drops to the small value once exact delay compensation has been achieved.}
\end{remark}
\subsection{Consensus of the distributed observers}
\textcolor{black}{In the following, we illustrate the distributed observers can achieve consensus, which is essential for establishing output synchronization of the heterogeneous MASs.}
Now, let $\xi(k) = [\xi_1^\mathrm{T}(k),\xi_2^\mathrm{T}(k),\ldots,\xi_N^\mathrm{T}(k)]^\mathrm{T}\in\mathbb{R}^{Nq}$, and $\tilde{\xi}(k) = \xi(k)-\mathbf{1}_N\otimes \upsilon(k) = [\tilde{\xi}_1^\mathrm{T}(k),\tilde{\xi}_2^\mathrm{T}(k),\ldots,\tilde{\xi}_N^\mathrm{T}(k)]^\mathrm{T}\in\mathbb{R}^{Nq}$ with $\tilde{\xi}_i(k) = \xi_i(k) - \upsilon(k)\in\mathbb{R}^{q} $.  From \eqref{equ:distributed_ob_heter}, the following error system is obtained
\begin{align}
 	\tilde{\xi}(k+1)=&\ \bar{S} \tilde{\xi}(k) 
    + \beta (\mathcal{D}_0\otimes S) \psi_0(k) \notag\\ 
    & \textstyle  + \beta \sum_{i=1}^{N-1} (\mathfrak{a}_i\otimes S)\psi_i(k), \label{equ:tilde_xi_heter}
\end{align}
where $\otimes$ denotes the Kronecker product,
\begin{align*}
&\ \bar{S} =   \mathbf{I}_N\otimes S - \beta(\mathcal{H}+\mathcal{D}_0)\otimes S
= \big(\mathbf{I}_N-\beta(\mathcal{H}+\mathcal{D}_0)\big)\otimes S.
\end{align*}
Moreover, $\psi_0(k) \in\mathbb{R}^{Nq}$, $\psi_i(k)\in\mathbb{R}^{Nq}$ and $\mathfrak{a}_i \in\mathbb{R}^{N\times N}$, $i=1,2,\ldots,N-1$, satisfy
\begin{align*}
	\psi_0(k) =&\ \begin{bsmallmatrix}
		\tilde{\Upsilon}_1^{\mathrm{T}}(k-\tau_{0,1}) & \tilde{\Upsilon}_2^{\mathrm{T}}(k-\tau_{0,2}) & \cdots & \tilde{\Upsilon}_N^{\mathrm{T}}(k-\tau_{0,N})
	\end{bsmallmatrix}^{\mathrm{T}},\\
	\psi_i(k) =&\ \begin{bsmallmatrix}
		 \mathbf{0}_{iq}^{\mathrm{T}} & \tilde{\Xi}_{i,i+1}^{\mathrm{T}}(k-\tau_{i,i+1}) & \tilde{\Xi}_{i,i+2}^{\mathrm{T}}(k-\tau_{i,i+2})& \cdots & \tilde{\Xi}_{i,N}^{\mathrm{T}}(k-\tau_{i,N})
	\end{bsmallmatrix}^{\mathrm{T}},\\
	\mathfrak{a}_i =&\ \mathrm{diag}[\mathbf{0}_{i}^{\mathrm{T}}, a_{i+1,i}, a_{i+2,i}, \ldots, a_{N,i}].
\end{align*}

\begin{lemma} \label{lemma2}
	Let Assumption \ref{assum1} hold. If the coupling gain $\beta$ satisfies 
	\begin{align}
		\rho(S)\rho(\mathbf{I}_{N}-\beta(\mathcal{H}+\mathcal{D}_0))< 1, \label{ineq: beta1}
	\end{align}
	then, for all initial condition $ \upsilon(0)\in\mathbb{R}^q$, $\xi_i(0)\in\mathbb{R}^q$, $i=1,2,\ldots,N$, \\
	(i) $\tilde{\xi}(k)$ is bounded for $k\in[0,T_{N}]$;\\
	(ii) The distributed observers \eqref{equ:distributed_ob_heter} achieve consensus, $\lim_{k\to\infty}|\xi(k) - \mathbf{1}_N\otimes\upsilon(k)|=  0 $. Furthermore, $\lim_{k\to\infty}|\xi_{i}(k)-\xi_j(k)| = 0$, for any $i,j\in\{1,2,\ldots,N\}$.
\end{lemma}
\begin{proof}
	\textcolor{black}{First, we show that the compensation terms in \eqref{equ:tilde_xi_heter} vanish for all $k\ge T_N$. Consider any active edge $(i,r)$ with $a_{r,i}>0$, where $i<r$. Since $a_{\ell,\ell-1}>0$ for $\ell=2,3,\ldots,N$, there exists a directed path from node $r$ to node $N$; when $r=N$, this path is understood as the trivial path. By the definition of $T_i$, there exists a path from node $0$ to node $i$ whose weight is $T_i$. Concatenating this path with the edge $(i,r)$ and then with a path from node $r$ to node $N$ gives a path from node $0$ to node $N$. Hence, $T_N\ge T_i+\tau_{i,r}.$
    Therefore, for every $k\ge T_N$, we have $k-\tau_{i,r}\ge T_i$. By Lemma \ref{lemma1}, it follows that
    \begin{align}
        \tilde{\Xi}_{i,r}(k-\tau_{i,r})=0,
    \end{align}
    for every active edge $(i,r)$ with $a_{r,i}>0$.}
    
    \textcolor{black}{Similarly, if $a_{i,0}>0$, then the edge $(0,i)$ followed by a path from node $i$ to node $N$ forms a path from node $0$ to node $N$. Thus, $T_N\ge \tau_{0,i}$. Consequently, for every $k\ge T_N$,
    we have $k-\tau_{0,i}\ge T_0$. Again by Lemma \ref{lemma1}, one has
    \begin{align}
    	\tilde{\Upsilon}_{i}(k-\tau_{0,i})=0.
    \end{align}
    If $a_{i,0}=0$, the corresponding term is multiplied by the zero entry of $\mathcal D_0$. Hence, for all $i=1,2,\ldots,N-1$ and all $k\ge T_N$, $(\mathcal D_0\otimes S)\psi_0(k)=0$ and $(\mathfrak a_i\otimes S)\psi_i(k)=0$. Therefore, \eqref{equ:tilde_xi_heter} reduces to
    \begin{align}
    	\tilde{\xi}(k+1)=\bar S\tilde{\xi}(k),\quad k\ge T_N. \label{equ:tilde_xi_after_TN}
    \end{align}}
    \textcolor{black}{\indent Next, we show the boundedness of $\tilde{\xi}(k)$ on the transient interval $[0,T_N]$. Let $R_0=|\upsilon(0)|+|\xi(0)|$.
    Since the interval $[0,T_N]$ contains only finitely many time instants, we can repeatedly apply the linear recursions \eqref{equ:distributed_ob_heter} and \eqref{equ:predictor_upsilon_heter}--\eqref{equ:distributed_predictor_xi_s_heter}. Together with the zero-prehistory convention for negative time indices, each variable involved in these recursions over $[0,T_N]$ is a finite linear combination of $\upsilon(0)$ and $\xi_i(0)$, $i=1,2,\ldots,N$. Therefore, there exists a constant $C_T>0$ such that, for all admissible indices and all $0\le k\le T_N$, $|\upsilon(k)|+|\xi(k)|+|\Upsilon_i(k)|+|\Upsilon_{i,s}(k)| +|\Xi_{i,r}(k)|+|\Xi_{i,r,s}(k)| \le C_T(|\upsilon(0)|+|\xi(0)|)$.
    Indeed, the claim holds at $k=0$ because all initial values are finite and all negative-time predictor values are set to zero. If it holds up to time $k$, then \eqref{equ:distributed_ob_heter} and \eqref{equ:predictor_upsilon_heter}--\eqref{equ:distributed_predictor_xi_s_heter} imply that all variables at time $k+1$ are finite linear combinations of previously bounded variables. Since the number of steps is finite, the same type of bound holds on the whole interval $[0,T_N]$.}

    \textcolor{black}{Consequently, the prediction-error terms contained in $\psi_0(k)$ and $\psi_i(k)$ are bounded on $[0,T_N]$. Thus, there exists a constant $C_\psi>0$, for $0\le k<T_N$, such that
    $|\beta \sum_{i=1}^{N-1}(\mathfrak a_i\otimes S)\psi_i(k) + \beta(\mathcal D_0\otimes S)\psi_0(k)| \le C_\psi R_0$.
    From \eqref{equ:tilde_xi_heter}, for $0\le k\le T_N$, one has
    \begin{align}
    	\textstyle \tilde{\xi}(k) =&\ \textstyle \bar S^k\tilde{\xi}(0) + \sum_{\ell=0}^{k-1} \bar S^{k-1-\ell}
    	[ \beta \sum_{i=1}^{N-1}(\mathfrak a_i\otimes S)\psi_i(\ell)\notag\\
        & + \beta(\mathcal D_0\otimes S)\psi_0(\ell) ].
    \end{align}
    Since $k\le T_N$ and $T_N$ is finite, the matrices $\bar S^k$ and the finite sum above are bounded on this interval. Therefore, there exists a constant $M_T>0$ such that 
    \begin{align}
        |\tilde{\xi}(k)| \le M_T\big(|\upsilon(0)|+|\xi(0)|\big), \quad 0\le k\le T_N.
    \end{align}
    This proves (i).}
    
    \textcolor{black}{It remains to prove convergence. By the definition of $\bar S$,
    $\bar S = (\mathbf I_N-\beta(\mathcal H+\mathcal D_0))\otimes S$, and using the spectral-radius property of the Kronecker product gives
    \begin{align}
    	\rho(\bar S) = \rho\big(\mathbf I_N-\beta(\mathcal H+\mathcal D_0)\big)\rho(S).
    \end{align}
    Thus, condition \eqref{ineq: beta1} implies that $\bar S$ is Schur. It follows from \eqref{equ:tilde_xi_after_TN} that
    \begin{align}
    	\textstyle \lim_{k\to\infty}\tilde{\xi}(k)=0.
    \end{align}
    Since $\tilde{\xi}(k)=\xi(k)-\mathbf 1_N\otimes\upsilon(k)$, we have
    \begin{align}
    	\textstyle \lim_{k\to\infty}|\xi(k)-\mathbf 1_N\otimes\upsilon(k)| = 0.
    \end{align}
    Furthermore, for any $i,j\in\{1,2,\ldots,N\}$,
    \begin{align}
    	\textstyle |\xi_i(k)-\xi_j(k)| \le |\xi_i(k)-\upsilon(k)| + |\xi_j(k)-\upsilon(k)|.
    \end{align}
    Taking the limit on both sides yields
    \begin{align}
    	\textstyle \lim_{k\to\infty}|\xi_i(k)-\xi_j(k)|=0.
    \end{align}
    This proves (ii) and completes the proof of Lemma \ref{lemma2}.}
\end{proof}
\section{Distributed control design}\label{distributed-control}
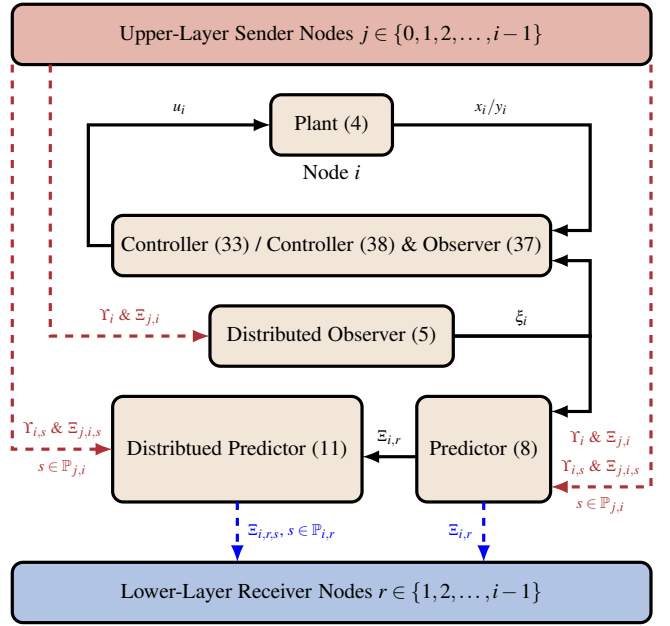
\begin{figure}[!t]
	\centering
	\begin{tikzpicture}[scale=0.8]
		
		\draw[rounded corners, very thick, fill=Maroon!30] (-4.25,1.5) rectangle (6.25,0.5) node[pos=.5] {\scriptsize Upper-Layer Sender Nodes $j \in\{0, 1,2,\ldots,i-1\}$};
		
		\draw[rounded corners, very thick, fill=brown!20] (0,0) rectangle (2,-1) node[pos=.5] {\scriptsize Plant \eqref{equ:plant_continuous}};
		\node[below] at (1, -1) {\scriptsize Node $i$};
		
		\draw[rounded corners, very thick,  fill=brown!20] (-2.6,-2) rectangle (4.6,-3) node[pos=.5] {\scriptsize Controller \eqref{equ:state_feedback_heter} / Controller \eqref{equ:output_feedback} \& Observer \eqref{equ:observer}};
		
		\draw[rounded corners, very thick, fill=brown!20] (-1,-3.5) rectangle (3,-4.5) node[pos=.5] {\scriptsize Distributed Observer \eqref{equ:distributed_ob_heter}};
		
		\draw[rounded corners, very thick, fill=brown!20] (2.4,-5) rectangle (4.6,-6.75) node[pos=.5] {\scriptsize Predictor \eqref{equ:predictor_xi_heter}};
		\draw[rounded corners, very thick, fill=brown!20] (-2.6,-5) rectangle (1.5,-6.75) node[pos=.5] {\scriptsize Distribtued Predictor \eqref{equ:distributed_predictor_xi_s_heter}};
		
		\draw[rounded corners, very thick, fill=NavyBlue!30] (-4.25,-7.75) rectangle (6.25,-8.75) node[pos=.5] {\scriptsize Lower-Layer Receiver Nodes $r\in\{i+1,i+2,\ldots,N\}$};
		
		\draw[->,very thick, >=latex]  (2, -0.5) -- node[midway,above]{\tiny $x_i / y_i$} (5.25,-0.5) -- (5.25,-2.25) ->  (4.6,-2.25) ;  
		\draw[->,very thick, >=latex]  (-2.6, -2.5) --   (-3,-2.5) -- (-3,-0.5) ->  node[midway,above]{\tiny $u_i$} (0,-0.5) ;  
		
		\draw[->,very thick, >=latex]  (3, -4) -- node[midway,above]{\tiny $\xi_i$} (5.25,-4) -- (5.25,-2.75) ->  (4.6,-2.75) ; 
		\draw[->,very thick, >=latex]  (3, -4) --  (5.25,-4) -- (5.25,-5.25) ->  (4.6,-5.25) ; 
		
		\draw[->,dashed, Maroon, very thick, >=latex]  (-3.625, 0.5) -- (-3.625,-4)  ->  (-1,-4) node[midway,  above]{\tiny $\Upsilon_i$ \&  $\Xi_{j,i}$}; 
		
		\draw[->, very thick, >=latex]  (2.4, -6)  -> (1.5, -6) node[midway,above]{\tiny $\Xi_{i,r}$}; 
		
		\draw[->,dashed, Maroon, very thick, >=latex, ]  (-4.25, 0.5) -- (-4.25,-5.875)  ->  (-2.6,-5.875) node[midway,  align=center,  above]{\tiny $\Upsilon_{i,s}$ \& $\Xi_{j,i,s}$}  node[midway,  align=center,  below]{\tiny $s \in \mathbb{P}_{j,i}$}; 
		\draw[->,dashed, Maroon,  very thick, >=latex]   (6.25, 0.5) -- (6.25,-6.5)  ->  (4.6,-6.5) node[midway,  align=center, above]{\tiny $\Upsilon_i$ \& $\Xi_{j,i}$ \\ \tiny  $\Upsilon_{i,s}$  \& $\Xi_{j,i,s}$ }  node[midway,  align=center, below]{\tiny $s \in \mathbb{P}_{j,i}$}  ; 
		
		\draw[->,dashed, blue, very thick, >=latex]  (-0.55, -6.75)  -> (-0.55, -7.75) node[midway, align=center, right]{\tiny $\Xi_{i,r,s}$, $s \in \mathbb{P}_{i,r}$}; 
		
		\draw[->,dashed, blue,  very thick, >=latex]  (3.5, -6.75)  -> (3.5, -7.75) node[midway, left]{\tiny $\Xi_{i,r}$}; 
	\end{tikzpicture}
	\caption{\textcolor{black}{Diagram of the prediction-based distributed control. Solid black lines indicate instantaneous intra-agent information flow. Dashed lines represent wireless transmission subject to communication delays: red lines denote information received from upper-layer nodes (delayed by $\tau_{j,i}$), while blue lines denote information sent to lower-layer nodes (delayed by $\tau_{i,r}$). In particular, to avoid redundancy, duplicate data (such as $\Upsilon_i$, $\Xi_{j,i}$, $\Upsilon_{i,s}$ and $\Xi_{j,i,s}$) shared across different blocks is transmitted to the central processor only once. The processor subsequently distributes the specific information required by each respective block.}}
	\label{fig_control}
\end{figure}
In this section, we provide two distributed control designs to achieve the output synchronization of MASs with communication delays. In Fig. \ref{fig_control}, we exhibit the diagram of the prediction-based distributed control.

\subsection{Distributed state feedback}
In the following, we first present a prediction-based distributed state feedback design to deal with the output synchronization problem of  MASs with distinct communication delays.
Firstly, we provide a technical lemma (borrowed from \cite{huang2016cooperative}) to illustrate the main results, which is specified as the following:
\begin{lemma} \label{lemma3}
	Consider the system $x(k+1) = Ax(k) + f(k)$,
	where $x(k)\in\mathbb{R}^n$, $A\in\mathbb{R}^{n\times n}$ is Schur, $f(k)$ is well defined for all $k\in \mathbb{N}$. If $f(k)\to 0$ as $k\to\infty $, then, for any $x(0)\in\mathbb{R}^n$, $x(k) \to 0$ as $k\to\infty$.
\end{lemma}
For $i=1,2,\ldots,N$, we design the following distributed state feedback  controller:
\begin{align}
	u_i(k) = K_{x_i}x_i(k) + K_{\xi_i}\xi_i(k),\ i=1,2,\ldots,N, \label{equ:state_feedback_heter}
\end{align}
where $K_{x_i}\in\mathbb{R}^{m_i\times n_i}$ is feedback gain such that $A_i+B_iK_{x_i}$ is Schur, the existence of which is guaranteed by Assumption \ref{assum3}. Furthermore, feedforward gain $K_{\xi_i}\in\mathbb{R}^{m_i\times q}$ is defined as  $K_{\xi_i}=U_i-K_{x_i}X_i$ with $(X_i,U_i)$ being the solution pair to the regulator equations \eqref{equ:regulation}. 
\begin{theorem} \label{theorem1}
	Let Assumptions \ref{assum1}--\ref{assum4} hold. Consider the closed-loop system consisting of the multi-agent system \eqref{equ:plant}, the control laws \eqref{equ:state_feedback_heter}, and the distributed observers  \eqref{equ:distributed_ob_heter} with the predictors \eqref{equ:predictor_upsilon_heter}, \eqref{equ:predictor_xi_heter} and distributed predictors \eqref{equ:distributed_predictor_upsilon_s_heter}, \eqref{equ:distributed_predictor_xi_s_heter}. Then, there exist  feedback gain $K_{x_i}\in\mathbb{R}^{m_i\times n_i}$ such that $A_i+B_iK_{x_i}$ is Schur, feedforward gain $K_{\xi_i}\in\mathbb{R}^{m_i\times q}$ satisfies $K_{\xi_i}=U_i-K_{x_i}X_i$ with $(X_i,U_i)$ being the solution of \eqref{equ:regulation}, $i=1,2,\ldots,N$, and coupling gain $\beta\in\mathbb{R}$ satisfies the condition \eqref{ineq: beta1}. Furthermore,  the output synchronization problem of heterogeneous multi-agent systems with distinct communication delays is solved.
\end{theorem}
\begin{proof}
	Let regulated output $e_i(k) = C_ix_i(k) + F\upsilon(k)$, regulated state $\tilde{x}_i(k) = x_i(k) - X_i\upsilon(k)$, and regulated control input $\tilde{u}_i(k) = u_i(k) - U_i\upsilon(k)$. By making use of the solution to  the regulator equations \eqref{equ:regulation}, it follows that, for $i=1,2,\ldots,N$
	\begin{align}
		\tilde{x}_i(k + 1)  
		=&\ A_i \tilde{x}_i(k)  + B_i \tilde{u}_i(k) + (A_i X_i  + B_iU_i  - X_iS)\upsilon(k) \notag\\
		=&\ A_i \tilde{x}_i(k)  + B_i \tilde{u}_i(k),  \label{equ:tilde_x} \\
		e_i(k) = &\ C_ix_i(k) + F\upsilon(k) 
		=  C_i\tilde{x}_i(k). \label{equ:e}
	\end{align}
	Besides, based on the definition of $\tilde{u}(k)$, it implies that
	\begin{align}
		\tilde{u}_i(k) 
		=&\ K_{x_i}\tilde{x}_i(k)  + K_{\xi_i}\tilde{\xi}_i(k) - (U_i - K_{x_i}X_i- K_{\xi_i}) \upsilon(k)\notag\\
		=&\ K_{x_i}\tilde{x}_i(k)  + K_{\xi_i}\tilde{\xi}_i(k). \label{equ:tilde_u}
	\end{align}
	Then, substituting \eqref{equ:tilde_u} into \eqref{equ:tilde_x} yields
	$\tilde{x}_i(k + 1)  = (A_i + B_iK_{x_i})\tilde{x}_i(k)  + B_iK_{\xi_i}\tilde{\xi}_i(k)$.
	Furthermore, using Lemmas \ref{lemma2}--\ref{lemma3}, we have $\lim_{k\to\infty}\tilde{x}_i(k)=0$, which follows that $\lim_{k\to\infty}e_i=\lim_{k\to\infty}C_i\tilde{x}_i(k)=0$,  $i=1,2,\ldots,N$. Then, it implies that $\lim_{k\to\infty}y_i(k)= - \lim_{k\to\infty} F\upsilon(k)=\lim_{k\to\infty}y_j(k)$, for any $i,j\in\{1,2,\ldots,N\}$. Therefore, the output synchronization problem is solved, i.e., $\lim_{k\to\infty}|y_i(k)-y_j(k)|=0$, $i,j\in\{1,2,\ldots,N\}$. The proof of Theorem \ref{theorem1} is completed.
\end{proof}
\subsection{Distributed dynamical output feedback}
In the sequel, we further consider the problem where the state of each agent is unmeasurable. Firstly, we present a observer design, and then provide a prediction-based dynamical output feedback design.

To overcome the unmeasurable state problem, inspired by \cite{huang2016cooperative}, for \textbf{Agent} $\boldsymbol{i}$, $\boldsymbol{i}=1,2,\ldots,N$, the following Luenberger observers are introduced:
\begin{align}
	\hat{x}_i(k+1) = A_i\hat{x}_i(k) + B_iu_i(k) + L_i(y_i(k) - C_i\hat{x}_i(k)),  \label{equ:observer}
\end{align}
in which $L_i$ denotes observer gain for \textbf{Agent} $\boldsymbol{i}$, chosen such that $A_i+L_iC_i$ is Schur. The existence of such a matrix is guaranteed by the detectability condition in Assumption \ref{assum3}.
Furthermore, based on \eqref{equ:observer}, the following prediction-based distributed dynamical output feedback  controller is given by
\begin{align}
	u_i(k) = K_{\hat{x}_i} \hat{x}_i(k) + K_{\xi_i}\xi_i(k),\ i=1,2,\ldots,N, \label{equ:output_feedback}
\end{align}
where $K_{\hat{x}_i} \in\mathbb{R}^{m_i\times n_i}$ is feedback gain such that $A_i + B_iK_{\hat{x}_i}$ is Schur, and the feedforward gain $K_{\xi_i}\in\mathbb{R}^{m_i\times q}$ satisfies $K_{\xi_i}=U_i-K_{\hat{x}_i}X_i$ with $(X_i,U_i)$ being the solution of \eqref{equ:regulation}.
\begin{theorem} \label{theorem2}
	Let Assumptions \ref{assum1}--\ref{assum4} hold. Consider the closed-loop system consisting of the multi-agent system \eqref{equ:plant}, the observer \eqref{equ:observer}, the control laws \eqref{equ:output_feedback} and the distributed observers  \eqref{equ:distributed_ob_heter} with the predictors \eqref{equ:predictor_upsilon_heter}, \eqref{equ:predictor_xi_heter} and distributed predictors \eqref{equ:distributed_predictor_upsilon_s_heter}, \eqref{equ:distributed_predictor_xi_s_heter}. Then, there exist  feedback gain $K_{\hat{x}_i}\in\mathbb{R}^{m_i\times n_i}$ such that $A_i+B_iK_{\hat{x}_i}$ is Schur, feedforward gain $K_{\xi_i}\in\mathbb{R}^{m_i\times q}$ satisfies $K_{\xi_i}=U_i-K_{\hat{x}_i}X_i$ with $(X_i,U_i)$ being the solution of \eqref{equ:regulation}, $i=1,2,\ldots,N$, and coupling gain $\beta\in\mathbb{R}$ satisfies the condition \eqref{ineq: beta1}. Furthermore,  the output synchronization problem of heterogeneous multi-agent systems with distinct communication delays is solved.
\end{theorem}
\begin{proof}
	Similar to the proof of Theorem \ref{theorem1}. We define observer error as $\breve{x}_i(k) = x_i(k)-\hat{x}_i(k)$. The definitions of regulated output $e_i(k)$,  regulated state $\tilde{x}_i(k)$, and regulated control input $\tilde{u}_i(k)$, $i=1,2,\ldots,N$ remain the same as those in the proof of Theorem \ref{theorem1}. Then, based on the definition, the observer error satisfies
		$\breve{x}_i(k + 1) = (A_i + L_iC_i) \breve{x}_i(k)$, $i=1,2,\ldots,N$.
	Owing to $(A_i + L_iC_i)$ being Schur, $\lim_{k\to\infty}\breve{x}_i(k) = 0$ exponentially. 
	Besides, according to the definition of $\tilde{u}_i(k)$, we have
	$\tilde{u}_i(k) = K_{\hat{x}_i}\tilde{x}_i(k) + K_{\xi_i}\tilde{\xi}_i(k) - K_{\hat{x}_i}\breve{x}_i(k)$.
	Furthermore, noting the fact that the regulated state  and the regulated output  still satisfy \eqref{equ:tilde_x}--\eqref{equ:e}, respectively, it follows that
	$	\tilde{x}_i(k + 1) = (A_i +  B_iK_{\hat{x}_i})\tilde{x}_i(k) + B_iK_{\xi_i}\tilde{\xi}_i(k) - B_iK_{\hat{x}_i}\breve{x}_i(k)$.
	Then, applying Lemmas \ref{lemma2}--\ref{lemma3}, we acquire $\lim_{k\to\infty}\tilde{x}_i(k)=0$. Moreover,  we have $\lim_{k\to\infty}e_i=\lim_{k\to\infty}C_i\tilde{x}_i(k)=0$,  $i=1,2,\ldots,N$. Then, we obtain $\lim_{k\to\infty}y_i(k)= - \lim_{k\to\infty} F \upsilon(k)$ $ = \lim_{k\to\infty}y_j(k)$, for any $i,j\in\{1,2,\ldots,N\}$. Therefore, the output synchronization problem is solved, i.e., $\lim_{k\to\infty}|y_i(k)-y_j(k)|=0$, $i,j\in\{1,2,\ldots,N\}$. The proof of Theorem \ref{theorem2} is completed.
\end{proof}
\section{Simulation results and koopman-based SIR application}\label{simulation}

\subsection{Numerical simulation}
This section provides a numerical example to evaluate the effectiveness of the proposed schemes. 
A multi-agent system composed of four agents is considered, with an exosystem serving as the leader. Specifically, the system matrices, input and output vectors of agents \eqref{equ:plant} are given as follows:
\begin{gather*}
	A_i = \begin{bNiceMatrix}[small,xdots/shorten=6pt]
		0 & 1 & 0 \\ 0 & 0 & \alpha_{i,1} \\ \alpha_{i,2} &  \alpha_{i,3} & \alpha_{i,4}
	\end{bNiceMatrix}, \  
	B_i = \begin{bNiceMatrix}[small,xdots/shorten=6pt]
		0 \\ 0 \\ 1 
	\end{bNiceMatrix},\ 
	C_i^\mathrm{T} = \begin{bNiceMatrix}[small,xdots/shorten=6pt]
		1 \\ 0 \\ 0
	\end{bNiceMatrix}.
\end{gather*}
The parameters $\{\alpha_{i,1}, \alpha_{i,2}, \alpha_{i,3}, \alpha_{i,4} \}$ are chosen as $\{1,0,1,1\}$, $\{2,0,1,2\}$, $\{1,3,1,2\}$ and $\{2,1,1,1\}$. In addition, the exosystem is given by \eqref{equ:exosystem} with 
\begin{gather*}
	S = \begin{bNiceMatrix}[small,xdots/shorten=6pt]
		\cos(\omega\pi) & \sin(\omega\pi)\\ -\sin(\omega\pi) & \cos(\omega\pi)
	\end{bNiceMatrix},\ \omega = 1.2.
\end{gather*}
The eigenvalues of $S$ are as follows:
\begin{align*}
	\sigma(S) = \{\cos(\omega\pi) + \sin(\omega\pi)\mathrm{i}, \cos(\omega\pi) - \sin(\omega\pi)\mathrm{i}\}.
\end{align*}
 Fig. \ref{fig_1} exhibits the structure of communication topology. Based on this, we provide the weighted adjacency matrix $\mathcal{A}$ and Laplace matrix $\mathcal{L}$ as  follows:
\begin{gather*}
	\mathcal{A} =\begin{bNiceMatrix}[small,xdots/shorten=6pt]
		0 & 0 & 0 & 0 & 0\\
		1 & 0 & 0 & 0 & 0\\
		1 & 1 & 0 & 0 & 0\\
		0 & 1 & 0 & 0 & 0 \\
		0 & 1 & 1 & 1 & 0 \\
	\end{bNiceMatrix},\ 
	\mathcal{L} =\begin{bNiceMatrix}[small,xdots/shorten=6pt]
	0 & 0 & 0 & 0 & 0\\
	-1 & 1 & 0 & 0 & 0\\
	-1 & -1 & 2 & 0 & 0\\
	0 & -1 & 0 & 1 & 0 \\
	0 & -1 & -1 & -1 & 3 \\
	\end{bNiceMatrix}.
\end{gather*}
To guarantee that $\bar{S}$ is Schur, we select $\beta = 0.25$. Additionally, according to \eqref{equ:regulation}, and setting $F = [0, -1]$, the solution to the regulator equations are as follows:
\begin{align*}
	X_i =&\ \begin{bNiceMatrix}[small,xdots/shorten=6pt]
			0 & 1 \\ -\sin(\omega\pi) & \cos(\omega\pi) \\ -\frac{1}{\alpha_{i,1}}\sin(2\omega\pi) & \frac{1}{\alpha_{i,1}}\cos(2\omega\pi)
	\end{bNiceMatrix},\\
	U_i^\mathrm{T} =&\ \begin{bNiceMatrix}[small,xdots/shorten=6pt]
	-\frac{1}{\alpha_{i,1}}\sin(3\omega\pi)+ \alpha_{i,3}\sin(\omega\pi) + \frac{\alpha_{i,4}}{\alpha_{i,1}}\sin(2\omega\pi) \\ 
	\frac{1}{\alpha_{i,1}}\cos(3\omega\pi) - \alpha_{i,2} - \alpha_{i,3}\cos(\omega\pi) - \frac{\alpha_{i,4}}{\alpha_{i,1}}\cos(2\omega\pi)
	\end{bNiceMatrix}.
\end{align*}
In simulation, we consider the distinct communication delays $\tau_{0,1} = 4$, $\tau_{0,2} = 5$, $\tau_{1,2} = 6$, $\tau_{1,3} = 11$, $\tau_{1,4} = 3$, $\tau_{2,4} = 10$ and $\tau_{3,4} = 12$. To compensate for these delays, we employ the predictors \eqref{equ:predictor_upsilon_heter}--\eqref{equ:predictor_xi_heter} and distributed predictors \eqref{equ:distributed_predictor_upsilon_s_heter}--\eqref{equ:distributed_predictor_xi_s_heter} with extended prediction horizons determined according to \eqref{equ:prediction_horizon_heter}. Specifically, for the exosystem with respect to Agent $1$, the extended prediction horizon  corresponds to the modified weight of the path $(1,2,4)$, yielding $\mathbb{P}_{0,1} = \{0,1,\ldots,\tau_{1,2}+\tau_{2,4}-2\}$. Similarly, we can deduce $\mathbb{P}_{0,2} = \{0,1,\ldots,\tau_{2,4}-1\}$, $\mathbb{P}_{1,2} = \{0,1,\ldots,\tau_{2,4}-1\}$ and $\mathbb{P}_{1,3} = \{0,1,\ldots,\tau_{3,4}-1\}$. Furthermore, to achieve the output synchronization of the multi-agent system, the feedback gains are chosen as $K_{x_1} = [-0.0313,	-0.9375,	-0.5000]$, $K_{x_2} = [0.0030,	-1.0550,	-1.4000]$, $K_{x_3} = [-2.9950,	-1.0950, $ $ -1.4500]$ and $K_{x_4} = [-1.0156,	-0.9688,	-0.5000]$. Subsequently, the feedforward gains can be calculated by the condition $K_{\xi_i}=U_i-K_{x_i}X_i$, $i=1,2,3,4$,  where $K_{\xi_1} = [1.3898,	0.2363]$, $K_{\xi_2} = [0.7932,	0.0143]$, $K_{\xi_3} = [1.5300,	0.0572]$ and $K_{\xi_4} = [0.6949, 	0.1182]$. In the subsequent simulations, all controller parameters and gains are kept the same as those used above.
\begin{figure}[!t]
	\centering
	\begin{tikzpicture}[scale=0.5]
		
		

        \node[draw, circle, fill=Maroon!50, very thick, minimum size=0.5cm] (a0) at (0,0) {\scriptsize $0$};
		
		\node[draw, circle, fill=NavyBlue!50, very thick, minimum size=0.5cm] (a1) at (2,0) {\scriptsize $1$};
		\node[draw, circle, fill=NavyBlue!50, very thick, minimum size=0.5cm] (a2) at (4,1.5) {\scriptsize $2$};
		\node[draw, circle, fill=NavyBlue!50, very thick, minimum size=0.5cm] (a4) at (6,0) {\scriptsize $4$};
		\node[draw, circle, fill=NavyBlue!50, very thick, minimum size=0.5cm] (a3) at (4,-1.5) {\scriptsize $3$};

        \draw[->, dashed, Blue, very thick, >=latex]
			(a0.east) to (a1.west);

		\draw[->, dashed, Blue, very thick, >=latex]
			(a0.north east) to (a2.west);

		\draw[->, dashed, Blue, very thick, >=latex]
			(a1.east) to  (a4.west);

		\draw[->, dashed, Blue, very thick, >=latex]
			(a1.north east) to  (a2.south west);

		\draw[->, dashed, Blue, very thick, >=latex]
			(a1.south east) to  (a3.north west);

		\draw[->, dashed, Blue, very thick, >=latex]
			(a2.south east) to  (a4.north west);

		\draw[->, dashed, Blue, very thick, >=latex]
			(a3.north east) to  (a4.south west);
        
        

        
        
		
		
		
	\end{tikzpicture}
	\caption{The communication topology (agent 0 is the exosystem).}
	\label{fig_1}
\end{figure}
\begin{figure}[!t]
	\centering
	\includegraphics[width=3in]{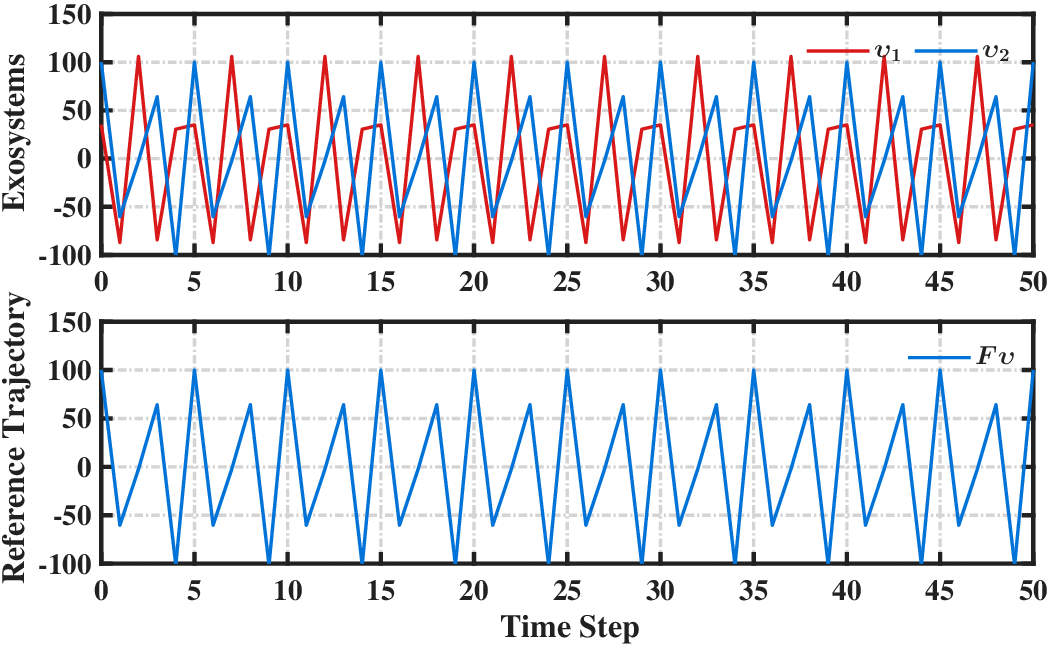}
	\caption{The evolution  of the exosystem.}
	\label{fig_2_50}
\end{figure}
\begin{figure}[!t]
    \centering
	\includegraphics[width=3.3in]{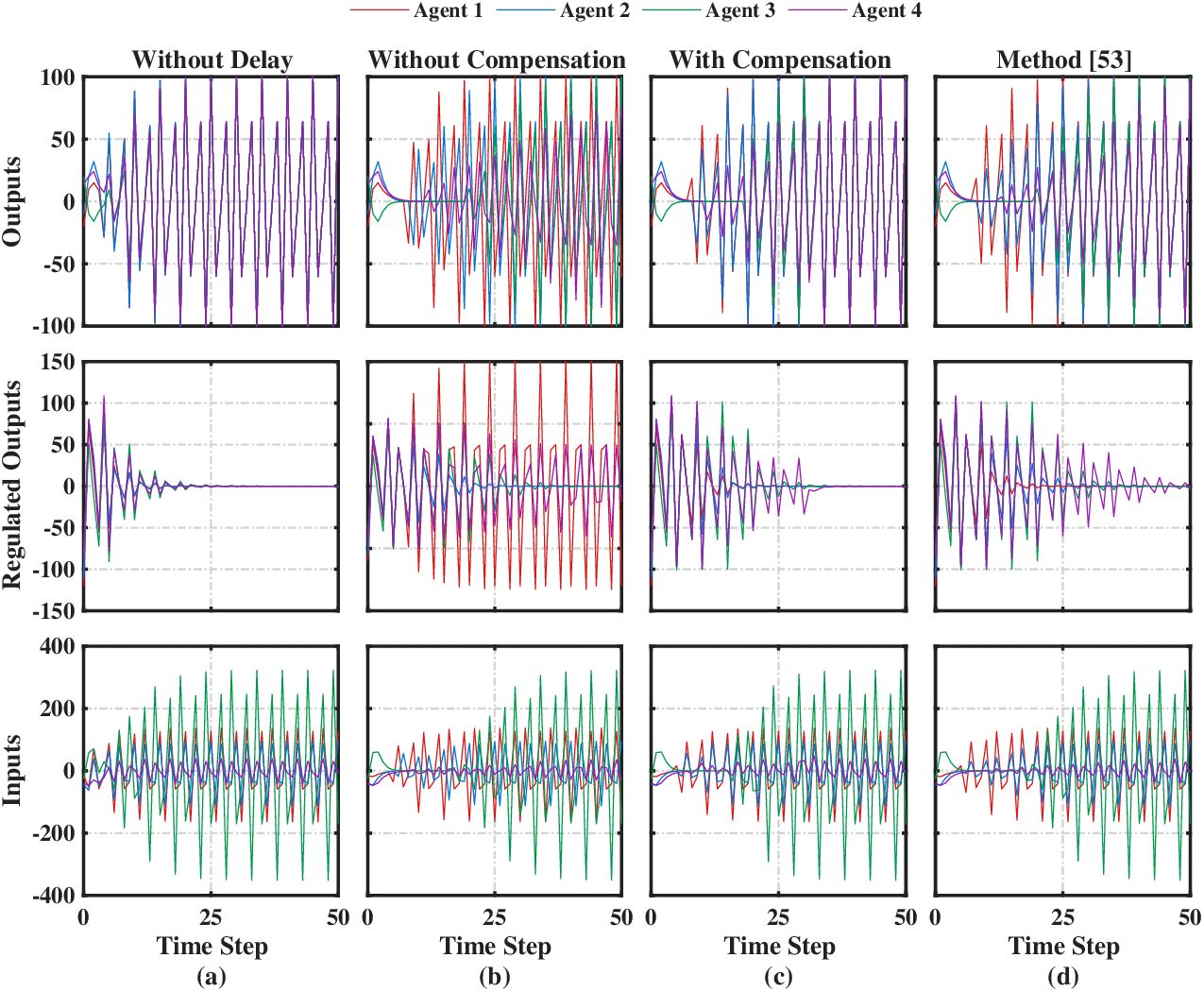} 
	\caption{Simulation results for the MAS \eqref{equ:plant} under the distributed state feedback. Top: Outputs; Middle: Regulated Outputs; Bottom: Inputs. (a)  Without communication delays (b) Without delay compensation (c) With delay compensation (d) Method in \cite{xu2017consensus}}
    \label{fig_3}
\end{figure}

The top subfigure of Fig.~\ref{fig_2_50} shows the evolution of the exosystem, and the control objective is to regulate all agent outputs to track the trajectory generated by the exosystem, as shown in the bottom subfigure of Fig.~\ref{fig_2_50}. For comparison, Fig.~\ref{fig_3}a presents the delay-free case under the distributed state feedback in \cite{huang2016cooperative}, where output synchronization is achieved. In contrast, Fig.~\ref{fig_3}b shows that communication delays without compensation significantly degrade the controller's performance and prevent synchronization. To overcome this, predictors and distributed predictors are introduced to construct the prediction-based distributed state feedback, and the corresponding outputs, regulated outputs, and control inputs are shown in Fig.~\ref{fig_3}c. Simulation results indicate that the proposed method effectively compensates for distinct communication delays, improves the performance of the standard distributed state feedback, and thereby successfully achieves output synchronizations. Fig.~\ref{fig_3}d further presents the results obtained by the method in \cite{xu2017consensus}. Although this method can achieve synchronization under communication delays, it does not fully eliminate their influence. By contrast, the proposed approach employs predictors that exactly offset the effect of delays after a finite number of steps, thereby yielding improved transient performance. It is also worth noting that the control inputs exhibit oscillatory profiles. This is because, by Theorem~\ref{theorem1}, the regulated input satisfies $\lim_{k\to\infty}\tilde u_i(k)=0$. Hence, $u_i(k)$ asymptotically follows $U_i \upsilon(k)$. Since $\upsilon(k)$ is oscillatory in this simulation, the control inputs also exhibit oscillatory behavior.

\subsection{Application to SIR epidemic model}
In this section, we would like to further emphasize the importance of the predictor through the SIR epidemic model. During an epidemic outbreak, reducing the peak of infection is crucial for preventing the exhaustion of limited medical resources (e.g., hospital beds and intensive care units). It ensures that the number of infected individuals remains within the healthcare capacity, thereby avoiding a potential collapse of the healthcare system.

Consider population movement between rural and urban areas in epidemic outbreak regions. Due to the typically superior medical facilities and healthcare policies in urban area, contrasted with the limited medical resources in rural regions, rural residents tend to travel to urban area for medical treatment at a migration rate $m_{r,u}$. However, this process of population mobility may increase contact with susceptible populations of urban area, potentially exacerbating the spread of the epidemic. However, for cost-efficiency considerations, the total count is then relayed to urban governments after a lag, thereby introducing a communication delay $D$ in the information flow. Under this background, according to  \cite{zhang2025analysis},  modified SIR models are constructed as follows:
	\begin{align*}
		\dot{s}_r(t) &=  - s_r(t) \beta_{r}(1-m_{r,u})i_r(t),\\
		\dot{i}_r(t) &= s_r(t) \beta_{r}(1-m_{r,u})i_r(t)  - \gamma_ri_r(t), \\
		\dot{r}_r(t) &= \gamma_ri_r(t),\\
		\dot{s}_u(t) &= - s_u(t) \beta_{u}( i_u(t) + m_{r,u}i_r(t-D)),\\
		\dot{i}_u(t) &= s_u(t)\beta_{u} (i_u(t) + m_{r,u}i_r(t-D)) - \tilde{\gamma}_u(t)i_u(t),\\
		\dot{r}_u(t) &=  \tilde{\gamma}_u(t)i_u(t),
	\end{align*}
	where $s_{\sigma}(t)$, $i_{\sigma}(t)$ and $r_{\sigma}(t)$ denote the fractions of susceptible, infected, and recovered populations in region $\sigma$, respectively, with $\sigma=r$ and $\sigma=u$ representing the rural and urban regions. $\tilde{\gamma}_u(t) = \gamma_{u} + \Gamma(t)$ is  a distributed feedback mitigation strategy that dynamically adjusts the recovery rate. This is implemented by providing effective medication, medical supplies, and healthcare workers to the population, where $\Gamma(t)$ serves as a state feedback controller.  Here, $\beta_{\sigma}$ and $\gamma_{\sigma}$ represent the average number of individuals with whom an infectious individual makes sufficient contact to transmit the infection and the  natural recovery rate for subpopulation $\sigma$ ($\sigma\in\{r,u\}$).  Subsequently, since $s_{\sigma}(t) + i_{\sigma}(t) + r_{\sigma}(t) = 1$, and considering our interest lies mainly in the infection dynamics, we derive the following discrete-time reduced-order model with the Euler method: 
	\begin{align*}
		i_r(k+1) &= i_r(k) + h(1 - i_r(k) - r_r(k)) \beta_{r}(1-m_{r,u})i_r(k) \notag \\
		& \quad  - h\gamma_ri_r(k), \\
		r_r(k+1) &= r_r(k) + h\gamma_ri_r(k),\\
		i_u(k+1) &= i_u(k) + h(1 - i_u(k) - r_u(k))\beta_{u} [i_u(k) \notag \\
		&\quad + m_{r,u}i_r(k-\tau)]  - h\tilde{\gamma}_u(k)i_u(k),\\
		r_u(k+1) &= r_u(k) + h\tilde{\gamma}_u(k)i_u(k),
	\end{align*}
	where $h$ is sampling interval and $\tau = D/h$. In  \cite{zhang2025analysis}, the following  distributed feedback control law was proposed: 
	\begin{equation*}
		\Gamma(k) = (1 - i_u(k) - r_u(k))\beta_u(1+m_{r,u}).
	\end{equation*}
	Substituting this feedback law into the above model yields the following closed-loop dynamics:
		\begin{align}
		i_r(k+1) &= [1- h\gamma_r  + h(1 - i_r(k) - r_r(k)) \beta_{r}(1-m_{r,u})] \notag\\
		&\quad \times i_r(k), \label{equ:SIR_ir} \\
		r_r(k+1) &= r_r(k) + h\gamma_ri_r(k), \label{equ:SIR_rr}\\
		i_u(k+1) &= (1 - h\gamma_u) i_u(k) + h(1 - i_u(k) - r_u(k))\beta_{u} m_{r,u}\notag \\
		&\quad \times  (i_r(k-\tau) - i_u(k) ),  \label{equ:SIR_iu}\\
		r_u(k+1) &= r_u(k) + h(\gamma_u + (1 - i_u(k) - r_u(k))\beta_u(1+m_{r,u})) \notag \\
		&\quad \times i_u(k).  \label{equ:SIR_ru}
	\end{align}
    The above dynamics exhibits a structure similar to that of the distributed observers (19); however, it is nonlinear. We therefore leverage Koopman operator theory to construct a finite-dimensional linear approximation of the nonlinear SIR dynamics. In particular, the rural subsystem is autonomous and is lifted by the standard Koopman operator \cite{koopman1931hamiltonian} . In contrast, the urban subsystem is driven by the delayed state of the rural subsystem, which is treated as an exogenous input. Hence, for the urban subsystem, we adopt a Koopman-with-inputs-and-control (KIC) theory \cite{proctor2018generalizing}.

    Let $\bar{x}_r = [i_r,r_r]^\mathrm{T}$ and $\bar{x}_u =[i_u,r_u]^\mathrm{T}$. Then, the closed-loop dynamics \eqref{equ:SIR_ir}--\eqref{equ:SIR_ru} can be compactly represented as
	\begin{align}
		\bar{x}_r(k+1) &= f_r(\bar{x}_r(k)), \label{equ: compact_r}\\ 
		\bar{x}_u(k+1) &= f_u(\bar{x}_u(k),\bar{u}(k)), \label{equ: compact_u}
	\end{align}
	where $f_r: \mathbb{R}^2 \to \mathbb{R}^2$ and  $f_u: \mathbb{R}^2\times\mathbb{R} \to \mathbb{R}^2$ are nonlinear mappings .The input $\bar{u}(k)$ is generated by the delayed rural subsystem, namely, $\bar{u}(k)=L\bar{x}_r(k-\tau)$ with $L=[1,0]$ for $k\ge\tau$, and $\bar{u}(k)=0$ otherwise.
	
    Let $\mathcal{M}_r$  and $\mathcal{M}_u$  denote  the smooth manifolds on which the rural and urban states evolve, respectively. For the autonomous rural subsystem, let \(\mathcal F_r\) be a vector space of scalar observables $\phi_r:\mathcal M_r\to\mathbb R$. According to \cite{williams2015data, korda2018convergence}, for any observables $\phi_r \in \mathcal{F}_r$, the standard Koopman operator $\mathcal K_r:\mathcal F_r\to\mathcal F_r$ is defined by
    \begin{gather*}
		(\mathcal{K}_{r}\phi_r)(\bar{x}_r) =  \phi_r (f_r(\bar{x}_r)) = (\phi_r \circ f_r)(\bar{x}_r),
	\end{gather*}
    For the input-driven urban subsystem, let \(\mathcal F_{u}\) be a vector space of scalar observables depending on both the urban state and the input, i.e., $\phi_u:\mathcal{M}_u\times\mathcal{U}\to\mathbb{R}$, where $\mathcal{U}$ denotes the input manifold. Since the input \(\bar u(k)\) is supplied by the rural subsystem, we do not seek to identify a map from \(\bar u(k)\) to \(\bar u(k+1)\). Hence, for any observables  $\phi_u \in \mathcal{F}_u$, we define the KIC operator \cite{proctor2018generalizing} by
    \begin{gather*}
		(\mathcal{K}_{u} \phi_u)(\bar{x}_u,\bar{u}) =  \phi_u(F_u(\bar{x}_u,\bar{u})) = (\phi_u\circ F_u)(\bar{x}_u,\bar{u}),
	\end{gather*}
    where the augmented mapping $F_u:\mathcal{M}_u\times \mathcal{U} \to \mathcal{M}_u\times \mathcal{U}$ is given by $F_u(\bar{x}_u,\bar{u}) = (f_u(\bar{x}_u,\bar{u}),\bar{u}_0)$. Here, $\bar{u}_0\in\mathcal{U}$ denotes the zero input. In the
present scalar-input case, $\bar{u}_0=0$.
	
	\textcolor{black}{Although $f_r$ and $f_u$ are nonlinear mappings, the induced operators $\mathcal{K}_r$ and $\mathcal{K}_u$ are linear on the corresponding spaces of observables. Indeed, for any $\phi_{r,1},\phi_{r,2}\in\mathcal{F}_r$ and any $a,b\in\mathbb{R}$, one has
    $\mathcal{K}_r(a\phi_{r,1}+b\phi_{r,2}) = (a\phi_{r,1}+b\phi_{r,2})\circ f_r = a(\phi_{r,1}\circ f_r)+b(\phi_{r,2}\circ f_r) = a\mathcal{K}_r\phi_{r,1} + b\mathcal{K}_r\phi_{r,2}$. Similarly, one has, for any $\phi_{u,1},\phi_{u,2}\in\mathcal{F}_u$ and any $a,b\in\mathbb{R}$, $\mathcal{K}_u(a\phi_{u,1}+b\phi_{u,2}) = (a\phi_{u,1}+b\phi_{u,2})\circ F_u = a(\phi_{u,1}\circ F_u)+b(\phi_{u,2}\circ F_u) = a\mathcal{K}_u\phi_{u,1} +b\mathcal{K}_u\phi_{u,2}$. Therefore, both the autonomous subsystem and the input-driven subsystem can be described as linear evolutions in the spaces of observables.}

	\textcolor{black}{While the Koopman-type operators are linear, they are generally infinite-dimensional, which makes their exact computation intractable for most nonlinear systems. Therefore, finite-dimensional approximations are required for practical implementation. For each \(\sigma\in\{r,u\}\), we approximate $\mathcal{K}_\sigma$ on a finite-dimensional observable subspace $\mathcal{F}_{\sigma}^{N_\sigma} := \mathrm{span}\{\psi_{\sigma,1},\psi_{\sigma,2},\ldots,\psi_{\sigma,N_\sigma}\}\subset \mathcal{F}_\sigma$, where \(\psi_{\sigma,i}\), \(i=1,2,\ldots,N_\sigma\), are selected linearly 
	independent dictionary functions. To unify the notation for the autonomous and input-driven subsystems, define
	\begin{align*}
		\zeta_r \coloneqq \bar{x}_r \in\mathbb{R}^{n_r},\quad  \zeta_u \coloneqq (\bar{x}_u,\bar{u})\in\mathbb{R}^{n_u},
	\end{align*}
	where $n_r = 2$ and $n_u = 3$. Then, for \(\sigma\in\{r,u\}\), the corresponding lifted observable vector is 
	given by
	\begin{align*}
		\psi_\sigma(\zeta_\sigma) =
        \begin{bsmallmatrix}
			\psi_{\sigma,1}(\zeta_\sigma) & \psi_{\sigma,2}(\zeta_\sigma) &	\cdots & \psi_{\sigma,N_\sigma}(\zeta_\sigma)
		\end{bsmallmatrix}^{\mathrm T}.
	\end{align*}
    Let $\Pr_\sigma: \mathcal{F}_\sigma\to\mathcal{F}_\sigma^{N_\sigma}$, $\sigma\in\{r,u\}$, be the projection operator onto $\mathcal{F}_\sigma^{N_\sigma}$, and let $\mathcal{I}_\sigma:\mathcal{F}_\sigma\to\mathcal{F}_\sigma$ denote the identity operator on $\mathcal{F}_\sigma$. For any observable $\bar{\phi}_\sigma\in\mathcal{F}_\sigma^{N_\sigma}$, there exists a coefficient vector $\boldsymbol{a}_\sigma=[a_{\sigma,1},a_{\sigma,2},\ldots,a_{\sigma,N_\sigma}]^\mathrm{T}$ such that
	\begin{align*}
		\bar{\phi}_\sigma(\zeta_\sigma) =	\boldsymbol{a}_\sigma^\mathrm{T}\psi_\sigma(\zeta_\sigma),\  \sigma\in\{r,u\},
	\end{align*}
	Since $\mathcal{K}_\sigma\bar{\phi}_\sigma\in\mathcal{F}_\sigma$, the action of the Koopman-type operator can be decomposed as
	\begin{align*}
    	\textstyle 	\mathcal{K}_\sigma\phi_\sigma =	\Pr_\sigma	(\mathcal{K}_\sigma\bar{\phi}_\sigma)	+ (\mathcal{I}_\sigma-\Pr_\sigma)(\mathcal{K}_\sigma\bar{\phi}_\sigma),
	\end{align*}
	where $\Pr_\sigma(\mathcal{K}_\sigma\bar{\phi}_\sigma)$ is the finite-dimensional projected component, and $ (\mathcal{I}_\sigma-\Pr_\sigma)(\mathcal{K}_\sigma\phi_\sigma)$ is the projection residual. Since $\Pr_\sigma(\mathcal{K}_\sigma\bar{\phi}_\sigma)\in\mathcal{F}_\sigma^{N_\sigma}$, there exists a coefficient vector 
	$\boldsymbol{b}_\sigma=[b_{\sigma,1},b_{\sigma,2},\ldots,b_{\sigma,N_\sigma}]^\mathrm{T}$ such that
	\begin{align*}
		\textstyle \Pr_\sigma\left(\mathcal{K}_\sigma\bar{\phi}_\sigma\right)(\zeta_\sigma)	= \boldsymbol{b}_\sigma^\mathrm{T}	\psi_\sigma(\zeta_\sigma).
	\end{align*}
	Moreover, there exists a matrix $K_\sigma\in\mathbb{R}^{N_\sigma\times N_\sigma}$ such that $\boldsymbol{b}_\sigma = K_\sigma^\mathrm{T}\boldsymbol{a}_\sigma$.	Equivalently, the projected action of the Koopman-type operator can be written as
	\begin{align*}
		\textstyle \Pr_\sigma \left(\mathcal{K}_\sigma\bar{\phi}_\sigma\right)(\zeta_\sigma)
		=	\boldsymbol{a}_\sigma^\mathrm{T}K_\sigma\psi_\sigma(\zeta_\sigma).
	\end{align*}
	In the following, we specify the dictionary functions and identify the finite-dimensional matrices $K_\sigma$, $\sigma\in\{r,u\}$, using the Extended Dynamic Mode Decomposition (EDMD) algorithm \cite{williams2015data,korda2018convergence}. For each subsystem, the dictionary set $\mathcal{D}_\sigma$ is chosen as
	\begin{align*}
		\mathcal{D}_\sigma	= \{i_\sigma,r_\sigma\}\cup \{s_\sigma, s_\sigma i_\sigma, s_\sigma r_\sigma, i_\sigma r_\sigma, s_\sigma^2,i_\sigma^2,r_\sigma^2\} \cup \mathcal{O}_\sigma,
	\end{align*}
	where $s_\sigma=1-i_\sigma-r_\sigma$, $\mathcal{O}_r=\emptyset$, and
	$\mathcal{O}_u=\{s_u\bar{u},\bar{u}\}$. By stacking the elements of $\mathcal{D}_\sigma$ in the order given above, we obtain the lifted observable vector $\psi_\sigma$, where $N_r=9$ and $N_u=11$.}

	\textcolor{black}{Based on simulations of the dynamics \eqref{equ: compact_r}--\eqref{equ: compact_u}, we collect the snapshot sequences of the state and exogenous input
	\begin{align*}
		\bar{X}_\sigma = &\
        \begin{bsmallmatrix}
            \zeta_\sigma(1) & \zeta_\sigma(2) & \cdots & \zeta_\sigma(m) & 
	    \end{bsmallmatrix}\in \mathbb{R}^{n_\sigma\times m}, \  \sigma\in\{r,u\},\\
        \bar{U} = &\
        \begin{bsmallmatrix}
            \bar{u}(1) & \bar{u}(2) & \cdots & \bar{u}(m)
	    \end{bsmallmatrix}\in \mathbb{R}^{n_\sigma\times m}.
	\end{align*}
	Then, we construct the following lifted data matrices
	\begin{align*}
		\bar{Y}_{\sigma} &= \begin{bsmallmatrix}
			\psi_\sigma(\zeta_{\sigma }(1)) & \psi_\sigma(\zeta_{\sigma}(2)) &\cdots & \psi_\sigma(\zeta_{\sigma}(m-1))
		\end{bsmallmatrix}\in \mathbb{R}^{N_\sigma\times (m-1)},\\
		\bar{Y}_{\sigma}^+ &= \begin{bsmallmatrix}
			\psi_\sigma(\zeta_{\sigma }^+(1)) & \psi_\sigma(\zeta_{\sigma}^+(2)) &\cdots & \psi_\sigma(\zeta_{\sigma}^+(m-1))
		\end{bsmallmatrix}\in \mathbb{R}^{N_\sigma\times (m-1)},
	\end{align*}
    where $\zeta_r^+(k) \coloneqq f_r(\bar{x}_r(k)) = \bar{x}_r(k+1)$ and $\zeta_u^+(k) \coloneqq F_u(\bar{x}_u(k),\bar{u}(k)) = (\bar{x}_u(k+1),0)$.
	According to the following projected Koopman relation
	\begin{align*}
		\mathcal{K}_\sigma \bar{\phi}_\sigma(\bar{X}_\sigma) =&\  \boldsymbol{a}_\sigma^\mathrm{T}K_\sigma\bar{Y}_{\sigma} + \mathcal{R}_\sigma = 	\boldsymbol{a}_\sigma^\mathrm{T}\bar{Y}_{\sigma}^+,\ \sigma\in\{r,u\},
	\end{align*}
	where $\mathcal{R}_\sigma$ denotes the residual induced by the projection error and data 
	approximation. If the residual vanishes on the sampled data, then
		$K_\sigma\bar{Y}_{\sigma}	= \bar{Y}_{\sigma}^+$.
	In general, $\mathcal{R}_\sigma$ is nonzero, and $K_\sigma$ is obtained by minimizing the 
	residual in the least-squares sense:
	\begin{align*}
		\textstyle K_\sigma = \arg\min_{\widetilde{K}_\sigma}\|	\bar{Y}_{\sigma}^+	- \widetilde{K}_\sigma\bar{Y}_{\sigma}\|_F^2, \ \sigma\in\{r,u\}.
	\end{align*}
	Here,  $\|\cdot\|_F$ is the Frobenius norm. The corresponding analytical solution is $K_\sigma = \bar{Y}_{\sigma}^+\bar{Y}_{\sigma}^{\dagger}$, $\sigma\in\{r,u\}$,
	where $\dagger$ denotes the Moore--Penrose pseudoinverse \cite{jonathan2014on}. Consequently, the 
	finite-dimensional Koopman-type approximation is given by
	\begin{align*}
		\psi_\sigma(\zeta_\sigma^+(k))\approx	K_\sigma\psi_\sigma(\zeta_\sigma(k)),\quad \sigma\in\{r,u\}.
	\end{align*}}
	 Next, we adopt the Singular Value Decomposition (SVD) to compute the pseudoinverse of $\bar{Y}_{\sigma}$ ($\sigma\in\{r,u\}$), which follows that 
		$\bar{Y}_{\sigma} =\textstyle \textrm{U}_{\sigma} \Sigma_{\sigma} \textrm{V}_{\sigma}^\mathrm{H}$.
	 In particular, $A^\mathrm{H}$ represents the complex conjugate transpose of matrix $A$. Consequently, we obtain 
	 \begin{align*}
	 	K_\sigma = \textstyle \bar{Y}_{\sigma}^+\textrm{V}_{\sigma}\Sigma_{\sigma}^{-1}\textrm{U}_{\sigma}^\mathrm{H}, \quad \sigma\in\{r,u\}.
	 \end{align*}
	 Letting $A_r = K_r\in\mathbb{R}^{9\times 9}$, $C_r = [\mathbf{I}_2 | \mathbf{0}_{2\times 7} ]\in\mathbb{R}^{2\times 9}$ and $\bar{X}_r(k) = \psi_r(\bar{x}_r(k))$, we first derive the following Koopman operator-based linear SIR model for the rural area:
	 \begin{align}
	 	\bar{X}_r(k+1) =&\ A_r \bar{X}_r(k), \label{equ:Koopman_r}\\
	 	\bar{x}_r(k) =&\ C_r \bar{X}_r(k). \label{equ:Koopman_r_o}
	 \end{align}
	 Furthermore, the vector of observables $\psi_u$ is partitioned into a state component $\bar{X}_u$ and an input component $\bar{u}$, such that $\psi_u = [\bar{X}_u^\mathrm{T}, \bar{u}]^\mathrm{T}$. This leads to the following partitioned structure for the approximated Koopman operator:
	 \begin{align*}
	 	\begin{bsmallmatrix}
	 		\bar{X}_u(k+1)\\
	 		\bar{u}(k+1)
	 	\end{bsmallmatrix} = K_u\begin{bsmallmatrix}
	 	\bar{X}_u(k)\\
	 	\bar{u}(k)
	 	\end{bsmallmatrix}  = \begin{bsmallmatrix}
	 	A_u & B_u\\
	 	0 & 0
	 	\end{bsmallmatrix}\begin{bsmallmatrix}
	 	\bar{X}_u(k)\\
	 	\bar{u}(k)
	 	\end{bsmallmatrix},
	 \end{align*}
	 Moreover, we extract the dynamic part of the system
	 \begin{align*}
	 	\bar{X}_u(k+1) =&\ A_u \bar{X}_u(k) + B_u \bar{u}(k), \\
	 	\bar{x}_u(k) =&\ C_u \bar{X}_u(k), 
	 \end{align*}
	where $C_u = [\mathbf{I}_2 | \mathbf{0}_{2\times 8} ]\in\mathbb{R}^{2\times 10}$. Then, from the definition of $\bar{u}$, we have 
	 \begin{align}
	 	\bar{X}_u(k+1) =&\ A_u \bar{X}_u(k) + B_uLC_r \bar{X}_r(k-\tau), \label{equ:Koopman_u}\\
	 	\bar{x}_u(k) =&\ C_u \bar{X}_u(k). \label{equ:Koopman_u_o} 
	 \end{align}
 	Now, we have constructed the Koopman operator-based linear model \eqref{equ:Koopman_r}--\eqref{equ:Koopman_u_o}, which presents a structure similar to that of the distributed observers \eqref{equ:distributed_ob_heter}.
 	
 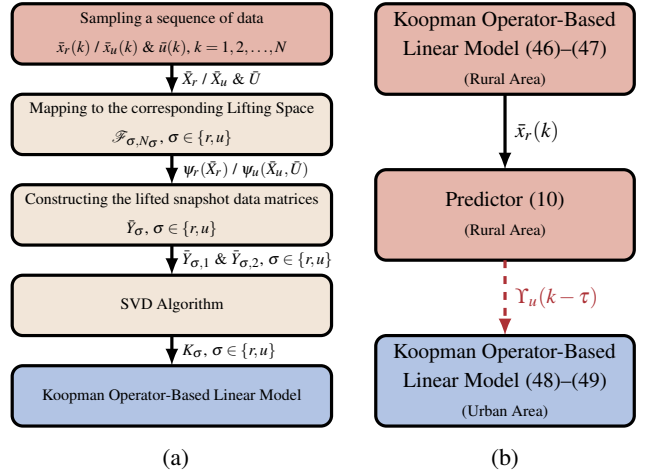
\begin{figure}
 		\centering
 		\subfloat[]{
 		\begin{tikzpicture}[scale=0.4]
 			
 			\draw[rounded corners, very thick, fill=Maroon!30] (-4.25,3) rectangle (6.25,1) node[pos=.5,align=center] {\tiny Sampling a sequence of data \\ \tiny $\bar{x}_r(k)$ / $\bar{x}_u(k)$ \& $\bar{u}(k)$, $k=1,2,\ldots,N$};
 			
 			\draw[->, very thick, >=latex]  (1, 1)  ->  (1,0) node[pos=.5,right]{\tiny $\bar{X}_{r}$ / $\bar{X}_{u}$ \& $\bar{U}$};
 			\draw[rounded corners, very thick,  fill=brown!20] (-4.25,0) rectangle (6.25,-2) node[pos=.5,align=center] {\tiny Mapping to the corresponding  lifting space \\ \tiny $\mathcal{F}_{\sigma,N_\sigma}$, $\sigma\in\{r,u\}$};
 			
 			\draw[->, very thick, >=latex]  (1, -2)  ->  (1,-3) node[pos=.5,right]{\tiny $\psi_r(\bar{X}_{r})$ / $\psi_u(\bar{X}_{u},\bar{U})$};
 			
 			\draw[rounded corners, very thick, fill=brown!20] (-4.25,-3) rectangle (6.25,-5) node[pos=.5,align=center] {\tiny Constructing the lifted snapshot data matrices \\ \tiny $\bar{Y}_{\sigma}$ \& $\bar{Y}_{\sigma}^+$, $\sigma\in\{r,u\}$}; 			
 			
 			\draw[->, very thick, >=latex]  (1, -5)  ->  (1,-6);  
 			
 			\draw[rounded corners, very thick, fill=brown!20] (-4.25,-6) rectangle (6.25,-8) node[pos=.5] {\tiny SVD algorithm}; 			
 			
 			\draw[->, very thick, >=latex]  (1, -8)  ->  (1,-9) node[pos=.5,right]{\tiny $K_{\sigma}$, $\sigma\in\{r,u\}$};
 			
 			\draw[rounded corners, very thick, fill=NavyBlue!30] (-4.25,-9) rectangle (6.25,-11) node[pos=.5] {\tiny Koopman operator-based linear model }; 
 		\end{tikzpicture} 	}
 	 \hfill
 	\subfloat[]{
 	\begin{tikzpicture}[scale=0.4]
 		
 		
 		\draw[rounded corners, very thick, fill=Maroon!30] (-3.25,9) rectangle (5.25,6) node[pos=.5, align = center] {\scriptsize Koopman Operator-Based\\ \scriptsize Linear Model \eqref{equ:Koopman_r}--\eqref{equ:Koopman_r_o} \\ \tiny (Rural Area)};
 		
 		\draw[->,very thick, >=latex]  (1, 6) ->  (1,3.5) node[pos=.5,right]{\scriptsize $\bar{x}_{r}(k)$};
 		
 		\draw[rounded corners, very thick, fill=Maroon!30] (-3.25,3.5) rectangle (5.25,0.5) node[pos=.5, align = center] {\scriptsize Predictor \eqref{equ:distributed_predictor_upsilon_s_heter} \\ \tiny (Rural Area)};

		\draw[->,dashed, Maroon, very thick, >=latex]  (1, 0.5) ->  (1, -2) node[pos=.5,right]{\scriptsize $\Upsilon_{u}(k-\tau)$}; 
		
		
 		\draw[rounded corners, very thick, fill=NavyBlue!30] (-3.25,-2) rectangle (5.25,-5) node[pos=.5, align = center] {\scriptsize Koopman Operator-Based\\ \scriptsize Linear Model \eqref{equ:Koopman_u}--\eqref{equ:Koopman_u_o} \\
 			\tiny (Urban Area)};
 	\end{tikzpicture}
 	}
 		\caption{Schematic diagram of the Koopman operator-based modeling and prediction framework. (a) Identification process of the linear model via EDMD algorithm; (b) Interconnected prediction structure between rural and urban areas with time delay (same-colored blocks indicate the same area).}
 		\label{fig_SIR}

 \end{figure}
 \begin{figure}[!t]
 	\centering
 	\includegraphics[width=3in,keepaspectratio]{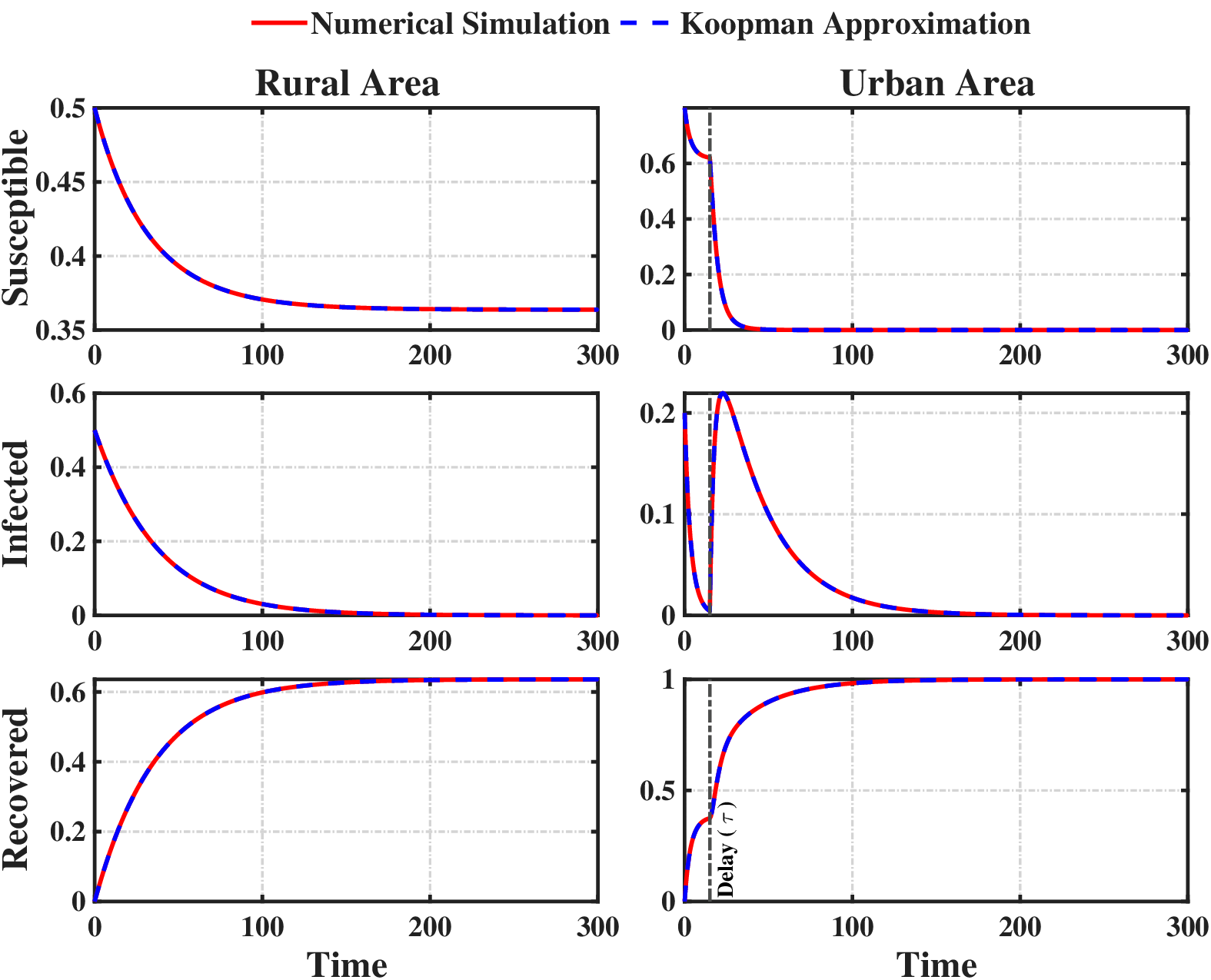}
 	\caption{Comparison of SIR model trajectories: Numerical simulation v.s. Koopman approximation.}
 	\label{fig_7}
 \end{figure}
 \begin{figure}[!t]
 	\centering
 	\includegraphics[width=3in]{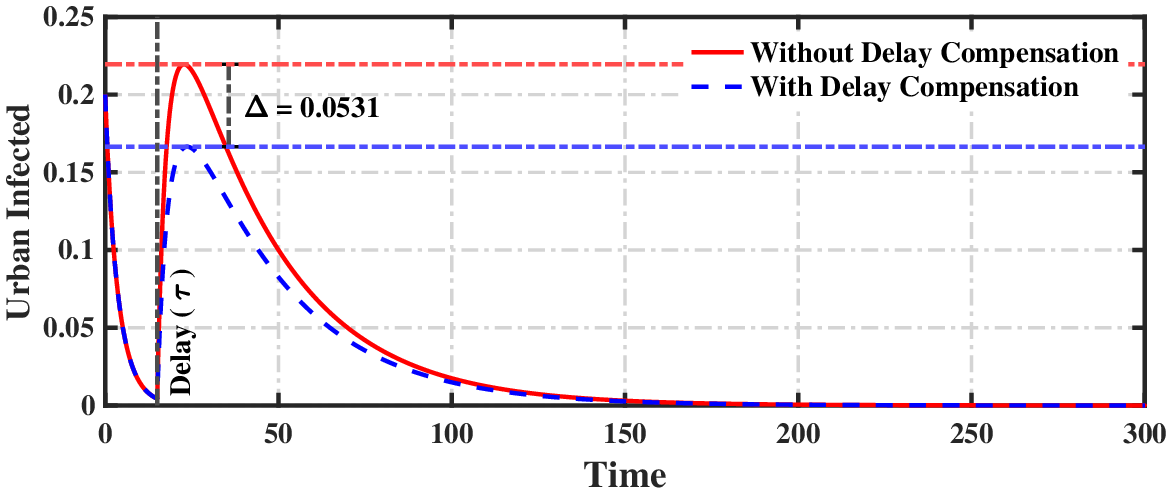}
 	\caption{Comparison of SIR model trajectories: Without delay compensation v.s. With delay compensation.}
 	\label{fig_8}
 \end{figure}
 
 	In the sequel, we present results from both numerical simulations and the Koopman operator-based linear approximation to verify the effectiveness of the constructed model. In simulation, we set the sampling interval $h = 0.01$. A relatively large time delay $D = 10$, coupled with a high mobility rate $m_{r,u} = 0.95$, is selected to accentuate the influence of the delay on the epidemic propagation. The system parameters are chosen as $\beta_r = \beta_u = 0.35$ and $\gamma_r = \gamma_u = 0.35$, which are adopted from \cite{zhong2021country, cooper2020sir} based on early COVID-19 data in China. The initial conditions are defined as $i_r(0) = 0.5$, $i_u(0) = 0.2$, $r_r(0) = r_u(0) = 0$. Fig.\ref{fig_7} presents the comparison between the numerical simulation results of the original nonlinear dynamics \eqref{equ:SIR_ir}--\eqref{equ:SIR_ru} and the Koopman operator-based linear approximation \eqref{equ:Koopman_r}--\eqref{equ:Koopman_u_o}. The trajectories of Koopman operator-based linear approximation (dashed blue lines) closely align with the ground truth of the original nonlinear system (solid red lines). This high fidelity confirms that the constructed finite-dimensional Koopman model successfully captures the dominant global dynamics. Subsequently, we apply the proposed delay compensation strategy to mitigate the communication delay, leveraging the Koopman operator-based linear model \eqref{equ:Koopman_r}--\eqref{equ:Koopman_u_o}. Since our primary focus is on the mitigation of the infected population, we solely present the trajectories of the infection. The comparison of simulation results is illustrated in Fig. \ref{fig_8}, where the solid red line represents the system without delay compensation, and the dashed blue line corresponds to the system with delay compensation. The difference in the peak values between the two scenarios is $\Delta = 0.0531$. While this margin may appear marginal due to population normalization, it translates to a substantial impact in real-world scenarios. For instance, in a city with a population exceeding $4$ million, the proposed delay compensation strategy ensures a reduction of over 200,000 infected individuals at the peak. Consequently, this reduction in infection numbers can, to a certain extent, decrease disease-induced mortality. This holds significant importance for effective epidemic prevention and control.
	 \begin{remark}
	 For the sake of computational simplicity, we use data from a single trajectory in this work. It is worth noting that the dataset can be constructed by concatenating multiple trajectories sampled from the manifold (see \cite{williams2015data, korda2018convergence}). This approach enriches the diversity of the data, thereby enhancing the generalization capability of the Koopman approximation model. Besides, since the Koopman operator provides a linear approximation of the underlying nonlinear dynamics, iterative application of the matrix $K_\sigma$ can lead to a gradual drift from the invariant manifold defined by the observables. To mitigate this error accumulation, we employ a re-lifting strategy. At each time step, instead of propagating the high-dimensional basis function vector $\psi_\sigma$ directly, we extract the physical state variables (e.g., $\bar{x}_\sigma$) from the linear prediction and re-evaluate the observable functions. This projection step restricts the trajectory to the valid state manifold and improves long-term prediction stability. Beyond standard EDMD, the Deep Koopman algorithm \cite{lusch2018deep, otto2019linearly} leverages deep learning to handle general nonlinear dynamics, particularly when suitable basis functions are unavailable or the underlying model is unknown.
	 \end{remark}
\section{Concluding remarks}\label{conclusion}
In this paper, we have investigated the output synchronization problem of discrete-time heterogeneous MASs subject to communication delays. To counteract the adverse effects of these delays, prediction-based distributed control strategies have been developed. Specifically, we integrated a predictor mechanism into the standard distributed observer design to handle communication delays. Building upon this modified distributed observer, we proposed prediction-based distributed state-feedback and dynamic output-feedback controllers to achieve output synchronization. Consequently, the outputs of all agents are regulated to track a common trajectory generated by an exosystem. Finally, the effectiveness of the proposed approaches has been verified through a numerical example and a Koopman operator-based linear SIR epidemic model. Numerical simulation results demonstrate that the proposed methods improve upon the standard distributed control scheme. In the context of the SIR model, our design effectively attenuates the infection peak.

%
%
\bibliographystyle{plain}        
\bibliography{autosam}           
\end{document}